\documentclass[11pt]{article}
\usepackage[margin=1in]{geometry}
\usepackage{multicol}

\usepackage{booktabs} % For formal tables

\usepackage[utf8]{inputenc}
\usepackage{authblk}
\usepackage{amsmath, amssymb, amsthm, thmtools, amsfonts, bm, bbm, thm-restate}
\usepackage{enumitem}
\usepackage{cite}
\usepackage[normalem]{ulem}
\renewcommand{\vec}[1]{\mathbf{#1}}

\allowdisplaybreaks

\usepackage{asymptote}
\usepackage{graphicx}
\usepackage{todonotes}
\usepackage{multirow}
\usepackage{comment}
\usepackage{dsfont}
\usepackage{bigstrut}
\usepackage{caption}
\usepackage{subcaption}
\usepackage{tcolorbox}
\usepackage{algorithm}
\usepackage{algorithmic}
\usepackage{graphicx}
\usetikzlibrary{patterns}

\usepackage{xcolor}
\definecolor{BrickRed}{rgb}{0.8,0.25,0.33}
\usepackage[pagebackref,colorlinks,citecolor=blue,linkcolor=BrickRed]
{hyperref}

\usepackage[framemethod=tikz]{mdframed}
\usepackage{nicefrac}

\usepackage{pifont} 

\usepackage{cleveref}

\theoremstyle{definition}

\theoremstyle{plain}
\newtheorem{theorem}{Theorem}[section]
\newtheorem{corollary}[theorem]{Corollary}

\newtheorem{fact}[theorem]{Fact}

\newtheorem{lemma}[theorem]{Lemma}

\newtheorem{claim}[theorem]{Claim}

\theoremstyle{definition}
\newtheorem{definition}{Definition}[section]

\crefname{thm}{Theorem}{theorems}
\crefname{cla}{Claim}{claims}
\crefname{lem}{Lemma}{lemmas}
\crefname{fact}{Fact}{facts}

\newcommand{\E}{\mathbb{E}}

\newcommand*{\todos}{}%

\usepackage[T1]{fontenc}
\usepackage{lmodern}

\ifdefined\notodos
\def\tristan#1{}
\def\david#1{}
\def\mahsa#1{}
\def\mark#1{}
\def\amin#1{}
\fi

\ifdefined\todos
\def\tristan#1{\marginpar{$\leftarrow$\fbox{T}}\footnote{$\Rightarrow$~{\sf\textcolor{cyan}{#1 --Tristan}}}}
\def\david#1{\marginpar{$\leftarrow$\fbox{D}}\footnote{$\Rightarrow$~{\sf\textcolor{blue}{#1 --David}}}}
\def\mahsa#1{\marginpar{$\leftarrow$\fbox{MD}}\footnote{$\Rightarrow$~{\sf\textcolor{purple}{#1 --Mahsa}}}}
\def\mark#1{\marginpar{$\leftarrow$\fbox{MB}}\footnote{$\Rightarrow$~{\sf\textcolor{green}{#1 --Mark}}}}
\def\amin#1{\marginpar{$\leftarrow$\fbox{A}}\footnote{$\Rightarrow$~{\sf\textcolor{red}{#1 --Amin}}}}
\fi

\usepackage{float}
\usepackage{wrapfig}
\restylefloat{figure}

\newcommand{\Opt}{\textup{OPT}}
\newcommand{\OptOn}{\textup{OPT}_{\textup{on}}}
\newcommand{\OptOff}{\textup{OPT}_{\textup{off}}}
\newcommand{\1}{\mathbbm{1}}
\newcommand{\Section}[1]{\section{#1}}
\newcommand{\Paragraph}[1]{\paragraph{#1.}}
\newcommand{\Subsection}[1]{\subsection{#1}}

\usepackage[square, numbers]{natbib}
\usepackage{cleveref}

\title{Optimal Rounding for Two-Stage Bipartite Matching}

\author{Tristan Pollner, Amin Saberi, Anders Wikum \thanks{Department of Management Science \& Engineering, Stanford University | \{\texttt{tpollner, saberi, wikum}\} @ stanford.edu}}

\begin{document}

\include{xspace}

\pagenumbering{gobble}
\date{}
\maketitle

\begin{abstract}  
We study two-stage bipartite matching, in which the edges of a bipartite graph on vertices $(B_1 \cup B_2, I)$ are revealed in two batches. In stage one, a matching must be selected from among revealed edges $E \subseteq B_1 \times I$. In stage two, edges $E^\theta \subseteq B_2 \times I$ are sampled from a known distribution, and a second matching must be selected between $B_2$ and unmatched vertices in $I$. The objective is to maximize the total weight of the combined matching. We design polynomial-time approximations to the optimum online algorithm, achieving guarantees of $\nicefrac{7}{8}$ for vertex-weighted graphs and $2\sqrt{2}-2 \approx 0.828$ for edge-weighted graphs under arbitrary distributions. Both approximation ratios match known upper bounds \cite{devanur2013randomized, naor2025online} on the integrality gap of the natural fractional relaxation, improving upon the best-known approximation of 0.767 by Feng, Niazadeh, and Saberi \cite{feng2021two} for unweighted graphs whose second batch consists of independently arriving nodes. 

Our results are obtained via an algorithm that rounds a fractional matching revealed in two stages, aiming to match offline nodes (respectively, edges) with probability proportional to their fractional weights, up to a constant-factor loss. We leverage negative association (NA) among offline node availabilities---a property induced by dependent rounding---to derive new lower bounds on the expected size of the maximum weight matching in random graphs where one side is realized via NA binary random variables. Moreover, we extend these results to settings where we have only sample access to the distribution. In particular, $\text{poly}(n,\epsilon^{-1})$ samples suffice to obtain an additive loss of $\epsilon$ in the approximation ratio for the vertex-weighted problem; a similar bound holds for the edge-weighted problem with an additional (unavoidable) dependence on the scale of edge weights.
\end{abstract}

\newpage
\pagenumbering{arabic}

\Section{Introduction}
In the two-stage bipartite matching problem, we are given a bipartite graph on vertices $(B_1 \cup B_2, I)$ whose edges are revealed in two batches. In the first stage, a matching must be selected from among edges $E \subseteq B_1 \times I$. In the second stage, edges $E^\theta \subseteq B_2 \times I$ are drawn from a known distribution $\mathcal{D}$, and a second matching is selected between $B_2$ and the unmatched vertices of $I$. The goal is to maximize the total weight of the final matching across both stages. In the vertex-weighted setting, each matched offline node $i \in I$ contributes weight $w_i \geq 0$, while in the edge-weighted setting, each matched edge $e \in E \cup E^\theta$ contributes weight $w_e \ge 0$. 

This problem is motivated by matching markets where participants are matched in batches, such as kidney exchange platforms \cite{yan2020dynamic} and ride-hailing services \cite{ferrari2015kidney}. As noted in prior work \cite[SODA'21]{feng2021two}, platforms often have partial information about future participants when making current matching decisions. For instance, ride-hailing apps such as Uber can incorporate knowledge of potential riders who are viewing prices but haven't yet confirmed a ride, when matching current riders to available drivers.

\cite{feng2021two} studied the two-stage bipartite matching problem in the context of adversarial and stochastic input. For the former, they give a randomized algorithm achieving an optimal competitive ratio of $\nicefrac{3}{4}$. In the stochastic setting introduced above, \cite{feng2021two} considered a model where all second-stage nodes $b \in B_2$ arrive with a known probability $\pi_b$, and achieved approximation guarantees of $1-1/e+1/e^2 \approx 0.767$ and $0.761$ for the unweighted and vertex-weighted versions, respectively. They additionally show that there is no FPTAS for this problem assuming $\textsf{P} \neq \textsf{NP}$.\footnote{In fact, it is possible to show the slightly stronger result that no PTAS exists assuming $\textsf{P} \neq \textsf{NP}$, via a construction of \cite{braverman2025new}; refer to \Cref{app:hardness}.}

\Subsection{Our Results}

In this work, we build on the results of Feng et al.\@ \cite{feng2021two} along three dimensions: we extend from unweighted to weighted graphs, from independently drawn vertices to arbitrary distributions, and improve the approximation guarantees to match the integrality gap. Specifically, we obtain a $\nicefrac{7}{8}$-approximation for vertex-weighted graphs and a $(2\sqrt{2} - 2)$-approximation for edge-weighted graphs, both in time polynomial in the number of nodes $n$ and the support size of $\mathcal{D}$.

\begin{theorem} \label{thm:mainvtx}
    There exists a $\textup{poly}(n, |\mathcal{D}|)$-time algorithm for two-stage bipartite achieving a $\nicefrac{7}{8}$-approximation to the optimum online algorithm, for vertex-weighted graphs. 
\end{theorem}

\begin{theorem} \label{thm:mainedge}
    There exists a $\textup{poly}(n, |\mathcal{D}|)$-time algorithm for two-stage bipartite achieving a $(2\sqrt{2}-2) \approx 0.828$-approximation to the optimum online algorithm, for edge-weighted graphs. 
\end{theorem}

Our approach differs from that of \cite{feng2021two} in that we develop new online algorithms for rounding batched \emph{fractional} matchings. In this closely related problem, a fractional matching is now revealed in two stages, and in each stage we must commit to an integral matching. The goal is match offline nodes (resp. edges) with probability proportional to their fractional matching weight, with a multiplicative loss at most some constant $c_{\textsf{V}}$ (resp. $c_{\textsf{E}}$). For both the vertex-weight and edge-weighted cases, we provide lower bounds matching previously known upper bounds of $c_{\textsf{V}} \le \nicefrac{7}{8}$ \cite{devanur2013randomized} and $c_{\textsf{E}} \le 2\sqrt{2} - 2$ \cite{naor2025online}. As a consequence, \Cref{cor:fractionalmatching} shows that our approximations in \Cref{thm:mainvtx} and \Cref{thm:mainedge} match the integrality gaps of the corresponding LP relaxations for the optimum online policy. 

\begin{restatable}{corollary}{fractionalmatchingbounds} \label{cor:fractionalmatching}
    For vertex-level (resp. edge-level) two-stage online fractional matching, the optimal rounding factor is $c_{\textsf{V}} = \nicefrac{7}{8}$ (resp. $c_{\textsf{E}} = 2\sqrt{2}-2$).  
\end{restatable}

\paragraph{Sample-based bounds.} A running time that scales polynomially in the support of distribution $\mathcal{D}$ may not be desirable for distributions with large support but compact representation---for example, the setting from \cite{feng2021two} where each second-stage online node appears independently with known probability. In addition, explicit knowledge of the distribution may not be realistic in practice. To handle both considerations, we further show that access to polynomially-many samples from $\mathcal{D}$ only degrades our approximation ratios by some additive $\epsilon$ factor. 

\begin{restatable}{corollary}{approxStatement}
    For any $\epsilon > 0$, given $k = \tilde{O}(n \cdot \epsilon^{-2})$ i.i.d. samples from $\mathcal{D}$ there exists a $\textup{poly}(n,k)$-time algorithm for the vertex-weighted two-stage bipartite problem with approximation ratio $\nicefrac{7}{8} - \epsilon$ to optimum online.
\end{restatable}

\begin{restatable}{corollary}{approxStatement}
    Given an instance of the edge-weighted two-stage bipartite problem, let $W$ be the max matching weight achievable over any graph in $\mathcal{D}$, and let $\mu$ be the expected gain of optimum online. For any $\epsilon > 0$, given $k = \tilde{O}\left(\frac{W^2}{\mu^2 \epsilon^2} \cdot n^2 \right)$ i.i.d. samples from $\mathcal{D}$, there exists a $\textup{poly}(n,k)$-time algorithm with approximation ratio $2\sqrt{2} - 2 - \epsilon$ to optimum online.
\end{restatable}

\subsection{Our Techniques}
Rounding the (fractional) solution to a linear programming relaxation is a standard technique in the online matching literature for developing approximation algorithms for the online optimum \cite{papadimitriou2021online, braverman2022max, naor2025online, braverman2025new}. Given a solution $(\vec{x},\vec{y})$ to an LP relaxation of the two-stage problem, our algorithm (\textsf{Round-Augment}) computes a first-stage matching by applying a dependent rounding procedure \cite{gandhi2006dependent} to the first-stage solution $\vec{x}$, then augments with a matching of maximum possible weight in the second stage. 
    
Our main technical contributions involve bounding the expected gain contributed by the second-stage max weight matching. One of the attractive properties of dependent rounding is that the procedure is known to introduce various forms of negative correlation \cite{gandhi2006dependent}. We prove a new dependence result: namely, that the (binary) availabilities of offline nodes for matching in the second stage satisfy negative association (NA). In light of this fact, for both edge-weighted and vertex-weighted bipartite graphs, we derive lower bounds on the expected size of the max weight matching in a random graph where one node set is randomly realized according to NA binary random variables. To establish these bounds, we develop two complementary analyses for the vertex-weighted and edge-weighted cases, both of which begin by reducing the problem to matching in acyclic bipartite graphs. 

In the vertex-weighted setting, we use submodular dominance \cite{qiu2022submodular} (as our objective is the rank function of a transversal matroid) to further reduce to independent node availabilities. Under this independence assumption, we prove a new result that generalizes the well-known lossless rounding of bipartite fractional matchings to lower bound the expected size of an unweighted max matching in the resulting random graphs. Notably, this result allows for random nodes on \emph{both} sides of the graph to assist our inductive proof. We then extend our result to vertex-weighted graphs by examining the matching decisions made by the (optimal) greedy algorithm for finding a max-weight independent set in the transversal matroid. Via strong duality, we show our global bounds imply rounding algorithms achieving balanced vertex-by-vertex guarantees. This final step is closely related to the existing literature on contention resolution schemes (CRS) \cite{chekuri2014submodular}; however, we require a type of CRS that is \emph{imbalanced} (i.e. not selecting items proportional to their probability of being available), departing from previous work.

In the edge-weighted setting, submodular dominance no longer applies. Rather than reducing to independent availabilities, we show that negative association of node availabilities can be preserved inductively when matching according to a particular \emph{monotone} contention resolution scheme. In detail, the algorithm uses the CRS to make matching decisions in a maximum-depth star subgraph before recursing on the remaining graph. We show that the existence of a contention resolution scheme with the desired properties on star graphs is a consequence of negative cylinder dependence. Moreover, we prove directly that this procedure produces a rounding scheme with balanced edge-by-edge guarantees.
 
\subsection{Further Related Work}

Online matching has been widely studied in the computer science literature since its introduction by \cite{karp1990optimal}, with extensions to vertex/edge-weighted graphs \cite{aggarwal2011online, fahrbach2020edge, gao2021improved, blanc2022multiway}. A long line of work has studied stochastic arrivals under competitive analysis (comparing to the optimum offline benchmark) \cite{feldman2009online, manshadi2012online, jaillet2013online, ezra2020online, huang2021online, tang2022fractional, huang2022power, chen2024stochastic}. The area is vast and we refer the reader to a recent book \cite{echenique2023online} and survey \cite{huang2024online} for a more complete treatment. 

Of particular relevance to our work is recent study on the effect of batching on online matching \cite{akbarpour2020thickness, feng2021two, feng2025batching}. Our contribution also fits into a line of work on two-stage stochastic combinatorial optimization problems, including on matching \cite{immorlica2004costs, shmoys2004stochastic, charikar2005sampling, shmoys2007approximation, katriel2008commitment, escoffier2010two, lee2020maximum, housni2021matching, feng2021two, jin2022online}. Among these works, \cite{feng2021two} is the most directly comparable, as it introduces the two-stage \emph{vertex-arrival} bipartite matching model we study. \cite{jin2022online} consider the same model under the ``algorithms with predictions'' framework, providing optimal robustness-consistency tradeoffs. Meanwhile, \cite{lee2020maximum} provides competitive ratio guarantees for two-stage \emph{edge-arrival} matching in bipartite and general graphs. 

Also relevant to our work is a recent interest in the optimum online benchmark. There is a growing body of results providing polynomial-time approximations and hardness results for discrete problems \cite{anari2019nearly, papadimitriou2021online, saberi2021greedy, feng2021two, braverman2022max, dutting2023prophet, braun2024approximating, dehaan2024matroid, naor2025online, braverman2025new} and continuous-time ``stationary'' settings \cite{aouad2020dynamic, kessel2021stationary, patel2024combinatorial, amanihamedani2024improved, amanihamedani2025adaptive}. This benchmark has also been studied in a recent line of work on online stochastic selection problems with unknown arrival order \cite{ezra2023next, ezra2023importance, chen2024setting, sun2025online}.

Our techniques are related to the literature on rounding and contention resolution. A sample of this work related to matching includes results for offline settings \cite{gandhi2006dependent, feige2006maximizing, chekuri2014submodular, bruggmann2020optimal, brubach2021improved, nuti2025towards}, online random-order settings \cite{adamczyk2018random, fu2021random, pollner2022improved, macrury2024random, macrury2025random}, and online adversarial order settings \cite{feldman2016online, lee2018optimal, ezra2020online, macrury2025random, naor2025online}. Our work gives the first tight bounds for rounding a fractional matching in two stages; a departure from the related work is the requirement of imbalanced guarantees. 
\Section{Model and Preliminaries} \label{sec:prelims}

We consider a random bipartite graph \(G^\theta \) with an offline node set $I$ and an online node set $B_1 \cup B_2$ partitioned into batches $B_1$ and $B_2$. Edges \(E \subseteq B_1 \times I\) incident to the first batch are deterministic, while edges $E^\theta \subseteq B_2 \times I$ incident to the second batch are realized according to some parameter $\theta \in \Theta$ drawn from a known distribution $\mathcal{D}$. We make no further assumptions on the distribution of $E^\theta$. Matching proceeds as follows: batches are revealed sequentially, and once a batch is revealed, we must irrevocably choose a matching between the online nodes in that batch and the offline nodes. Each offline node can be matched at most once across both batches. 

We use $\OptOn$ to denote the \emph{optimal online algorithm}, i.e., the online algorithm that achieves the largest matching value in expectation over $\mathcal{D}$. Here, matching value may refer to unweighted, vertex-weighted, or edge-weighted matching size and will be clear from context. If $\textup{ALG}(\theta)$ denotes the expected matching value returned by algorithm $\textup{ALG}$ on realization $\theta$, we say $\textup{ALG}$ is $\alpha$-approximate if for all distributions $\mathcal{D}$ we have 
\[\underset{\theta \sim \mathcal{D}}{\E}[\textup{ALG} (\theta)] \ge \alpha \cdot \underset{\theta \sim \mathcal{D}}{\E}[\OptOn(\theta)].\]

\Paragraph{Rounding a two-stage fractional matching} Our approximation guarantees are based on new algorithms for rounding a fractional matching that is revealed online in two stages. In this related problem, two fractional matchings $\vec x \in [0,1]^{B_1 \times I}$ and $\vec y \in [0,1]^{B_2 \times I}$ are revealed along with the edges in their respective batch. Naturally, the fractional mass on non-edges must be $0$. An online rounding scheme must commit to an integral first-stage matching given only knowledge of $\vec x$, then compute a second-stage matching in the remaining graph when $\vec y$ is revealed.  We assume $\vec{x}$ and $\vec{y}$ are chosen upfront and adversarially, with the guarantee that $\vec{x}+\vec{y}$ is a feasible fractional matching in the full graph.

To derive approximation guarantees in the vertex-weighted and edge-weighted settings, the corresponding goal in online rounding is to provide \emph{vertex-uniform} or \emph{edge-uniform} bounds, which match the marginals of $\vec x$ and $\vec y$ up to a multiplicative factor. A vertex-uniform bound is a constant $c_{\textsf{V}}$ such that for all offline nodes $i \in I$, $$\Pr[i \text{ matched}] \ge c_{\textsf{V}} \cdot \left( \sum_{a \in B_1} x_{ia} + \sum_{b \in B_2} y_{ib} \right).$$ Analogously, an edge-uniform bound is a constant $c_{\textsf{E}}$ such that for every edge $e$, $$\Pr[e \text{ matched}] \ge c_{\textsf{E}} \cdot \begin{cases} x_e & \text{if } e \text{ in first stage} \\
y_e & \text{if } e \text{ in second stage} \end{cases}.$$
Previous work has shown that the optimal factors for rounding a two-stage fractional matching are $c_{\textsf{V}} \le \nicefrac{7}{8}$ \cite{devanur2013randomized} and $c_{\textsf{E}} \le 2\sqrt{2} - 2$ \cite{naor2025online}.

\Paragraph{Dependent rounding and negative dependence} We construct two-stage matching algorithms from solutions to an LP relaxation by applying the dependent rounding scheme to the first-stage matching $\vec{x}$. Dependent rounding \cite{gandhi2006dependent} is an iterative procedure that performs the following operations to randomly round $\vec{x}$:
\begin{enumerate}
    \item Identify a cycle or maximal path $P$ consisting of non-integral edges in $\vec{x}$. Partition the edges of $P$ into matchings $M_1$ and $M_2$.
    \item Compute positive numbers $\alpha$ and $\beta$ based on $\vec{x}$, $M_1$, and $M_2$. With probability $\beta/(\alpha+\beta)$, update $x$ by setting $x_e \to x_e+\alpha$ for $e \in M_1$ and $x_e \to x_e-\alpha$ for $e \in M_2$. With probability $\alpha/(\alpha+\beta)$, update $x$ by setting $x_e \to x_e-\beta$ for $e \in M_1$ and $x_e \to x_e+\beta$ for $e \in M_2$.
\end{enumerate}
Upon termination, dependent rounding outputs a random integral matching $M$ with marginals $\vec{x}$ (i.e., $\Pr[e \in M] = x_e$ for all $e \in E$). Moreover, Gandhi et al. prove that the matching values on edges incident to any fixed offline node are negatively correlated (a non-trivial statement in their setting when vertices have integral capacities that may exceed 1). Our results rely on proving a different form of negative correlation, between nodes rather than edges. In particular, we prove that when a first-stage matching is computed via dependent rounding, the availabilities of offline nodes in the second stage exhibit \emph{negative association}.

\begin{definition}[Negative association]
Random variables $X_1, X_2, \dots, X_n \in [0,1]$ are \emph{negatively associated (NA)} if for any disjoint index sets $S, T \subseteq [n]$ and monotone increasing functions $f: [0,1]^{|S|} \to \mathbb{R}$ and $g: [0,1]^{|T|} \to \mathbb{R}$, we have \[\mathbb{E}[f(\vec{X}_S) \cdot g(\vec{X}_T)] \le \mathbb{E}[f(\vec{X}_S)] \cdot \mathbb{E}[g(\vec{X}_T)],\] where $\vec{X}_S$ denotes the collection $(X_i)_{i \in S}$. 
\end{definition}
NA random variables enjoy many useful properties, including that increasing functions each defined on disjoint subsets of NA random variables are themselves NA \cite{joag1983negative}. Negative association also implies the weaker conditions of \emph{negative cylinder dependence} and \emph{submodular dominance} \cite{qiu2022submodular}. 
A collection of Bernoulli random variables satisfies negative cylinder dependence if in any sub-collection, the probability that all variables take the same value is no larger than if the random variables were independent.
\begin{definition}[Negative Cylinder Dependence]
Bernoulli random variables $X_1, \dots, X_n$ are \emph{negatively cylinder dependent (NCD)} if for any $S \subseteq [n]$ we have \[\E \left[ \prod_{i \in S} X_i \right] \le \prod_{i \in S} \E[X_i] \quad \text{and} \quad \E \left[ \prod_{i \in S} \overline{X_i} \right] \le \prod_{i \in S} \E[\overline{X_i}].\]
\end{definition}
A distribution satisfies submodular dominance if its expectation under any submodular function is made no larger by replacing the distribution with independent samples from the marginals. In the following definition and throughout in the paper, we will use the notation $\tilde{\mathcal{D}}$ to denote the product distribution over the marginals of a distribution $\mathcal{D}$. Analogously, if $(X_i)$ is a collection of random variables whose joint distribution is given by $\mathcal{D}$, then $(\tilde{X}_i)$ is an independent collection with joint distribution $\tilde{\mathcal{D}}$.
\begin{definition}[Submodular dominance] A distribution $\mathcal{D}$ over $2^U$ satisfies \emph{submodular dominance} if for every  submodular function $f: 2^U \to \mathbb{R}$, \[\underset{S \sim \mathcal{D}}{\E}\left[f(S)\right] \geq \underset{S \sim \mathcal{\tilde{D}}}{\E}\left[f(S)\right].\]
\end{definition} 
\Section{LP Relaxation and Two-Stage Matching Algorithm}
In this section, we give our two-stage matching algorithm and formalize its structural properties.
\Paragraph{LP relaxation}
As in previous work, our two-stage matching algorithm is based on rounding the solution to a Linear Programming (LP) relaxation of the optimum online algorithm.  Our relaxation takes the following form:
\begin{align}
\max \quad & \sum_{a \in B_1} \sum_{i \in I} \tilde{w}_{ia} \cdot x_{ia} + \mathbb{E}_{\theta \sim \mathcal{D}}\Biggl[\sum_{b \in B_2} \sum_{i \in I} \tilde{w}_{ib} \cdot y^\theta_{ib}\Biggr] \tag{$\text{LP}_{\text{on}}$} \label{LPon} \\[1mm]
\text{s.t.} \quad & \sum_{a \in B_1} x_{ia} + \sum_{b \in B_2} y^\theta_{ib}  \le 1, \quad \forall i \in I, \forall \theta \in \Theta \label{eq:lp_online_constraint}\\[1mm]
& \sum_{i \in I} x_{ia} \le 1, \quad \forall a \in B_1, \label{eq:lp_A}\\[1mm]
& \sum_{i \in I} y^\theta_{ib} \le 1, \quad \forall b \in B_2, \forall \theta \in \Theta, \label{eq:lp_B}\\[1mm]
& x_{ia} \ge 0, \quad \forall i \in I, a \in B_1,\qquad y^\theta_{ib} \ge 0, \quad \forall i \in I, b \in B_2, \forall \theta \in \Theta. \label{eq:lp_nonneg}
\end{align}
For all \(i\in I\) and \(a\in B_1\), \(x_{ia}\) denotes a fractional matching variable corresponding to the probability of matching the first-stage edge $(i,a) \in E$. Similarly, for each possible realization $\theta \in \Theta$, the fractional variable \(y^\theta_{ib}\) corresponds to the probability of matching the second-stage edge $(i,b) \in E^{\theta}$, conditioned on realizing $\theta$. We enforce that non-edges receive zero weight, i.e.,
\begin{equation} x_{ia}=0 \quad \forall (i,a)\notin E \qquad \text{and} \qquad y^\theta_{ib}=0 \quad \forall (i,b)\notin E^{\theta}.\end{equation} We will define node-level fractional degrees by $x_i := \sum_{a \in B_1} x_{ia}$ and $y_i^{\theta} := \sum_{b \in B_2} y_{ib}^{\theta}$. This formulation captures unweighted ($\tilde{w}_{ia}=1$), vertex-weighted $(\tilde{w}_{ia} = w_i)$, and edge-weighted $(\tilde{w}_{ia} = w_{ia})$ variants with a suitable $\tilde{w}$. We note that \eqref{LPon} is a valid relaxation because an online algorithm's first-stage decisions must be independent of the realization $\theta$; for a complete proof we refer the reader to \Cref{app:deferredproofs}. We note, however, that \eqref{LPon} is not a valid relaxation of the the optimum offline benchmark, which can correlate its decisions in the first stage with the realization of $\theta$. For more discussion on the fractional relaxation of the optimum offline, and a simple $\nicefrac{3}{4}$-rounding, we refer the reader to \Cref{app:offline}.

\Paragraph{Algorithmic framework}
Our \textsf{Round-Augment} algorithm (\Cref{alg: round-augment}) for two-stage matching begins by applying dependent rounding to the first-phase LP solution $\vec{x}$, potentially scaled down by a constant $c$. Then, once the second-stage edges are realized, it augments the first-stage matching with a max weight matching over the available nodes in the second stage. As we will see in Sections \ref{sec: vtx-weight} and \ref{sec: edge-weight}, the scaling factor $c$ will depend on the particular variant of the problem, with $c = 1$ for the vertex-weighted algorithm and $c = 2\sqrt{2} - 2 \approx 0.828$ for the edge-weighted algorithm.
\begin{algorithm}[H]
\caption{\textsf{Round-Augment}}
\begin{algorithmic}[1]  \label{alg: round-augment}
\STATE \textbf{Input:} Optimal LP solution \((\vec{x},\vec{y})\) for \eqref{LPon}, scalar $c \in [0,1]$

\STATE Upon arrival of batch $B_1$:
\STATE \quad Round \(c \cdot \vec{x}\) using dependent rounding to obtain matching \(M_1\)
\STATE Upon arrival of batch \(B_2\) with realization \(\theta\)
\STATE \quad $
I_{\mathrm{avail}} \gets \{ i \in I : i \text{ is unmatched in } M_1 \} $

\STATE \quad $M_2 \gets$ max weight matching on graph induced by \(I_{\mathrm{avail}}\) and batch \(B_2\) (using the edges \(E^{\theta}\))
\STATE \textbf{Return:} \(M = M_1 \cup M_2\)
\end{algorithmic}
\end{algorithm}

Because dependent rounding preserves marginals, the contribution from the first stage matching is simple to compute. Our analysis hence reduces to lower bounds on the expected weight of the maximum second-stage matching.

\Paragraph{Second-stage structure} For any fixed realization $\theta \in \Theta$, randomness in the second stage comes from the (random) availability of offline nodes after applying dependent rounding. As such, it will be useful to formalize this type of random node-induced graph via the following notation. 
\begin{definition}
    Given a bipartite graph $G = (V, E)$ and a collection of indicator random variables $A=(A_v)_{v \in V}$ for whether node $v \in V$ is active, $G_A$ denotes the (random) subgraph of $G$ induced by active nodes. (Any indicators that are not explicitly specified are assumed to be active with probability 1.)
\end{definition}

For any fixed second-stage graph $G^\theta$, we are interested in computing a max matching in the random graph $G^\theta_A$, where $A_i = \mathbbm{1}\{i \in I_{\textup{avail}}\}$ indicates  whether offline node $i \in I$ remains available to be matched after dependent rounding. It is natural to expect that $(A_i)_{i \in I}$ exhibit some form of negative dependence due to contention in the matching process. \Cref{lemma: dep-round-NA} formalizes this intuition; when $(L, R) = (B_1, I)$, the result shows that negative association of $(A_i)_{i \in I}$ is a property of the dependent rounding procedure. As the proof is largely mechanical, we leave it to \Cref{app:deferredproofs}.

\begin{restatable}{lemma}{depRoundNA} \label{lemma: dep-round-NA}
    Let $G = (L, R, E)$ be a bipartite graph supporting a fractional matching $\vec{x}$. If $A_v$ is an indicator for whether node $v \in R$ remains available after applying dependent rounding to $\vec{x}$, then the random variables $(A_v)_{v \in R}$ are NA.
\end{restatable} 
\Section{Vertex-weighted} \label{sec: vtx-weight}
% \paragraph{Integrality gap.}
The vertex-weighted variant of \eqref{LPon} has an integrality gap of at most $\nicefrac{7}{8}$, given by the 8-cycle example of \cite{devanur2013randomized}, which is shown in \Cref{fig:tight-gap} for completeness. In this instance, vertices have unit weight and $E^{\theta}$ completes an 8-cycle in two different ways, each with probability $\nicefrac{1}{2}$. The optimal solution to \eqref{LPon} assigns a fractional value of $\nicefrac{1}{2}$ to every edge for an objective value of 4. However, no online algorithm can guarantee a matching whose expected size is greater than $\nicefrac{7}{2}$.
\begin{figure}[h!]
\caption{Instance with an integrality gap of $7/8$. Dashed edges are revealed in the second stage, and colors show parity in each 8-cycle.}
\label{fig:tight-gap}
\centering
\includegraphics[width=0.85\textwidth]{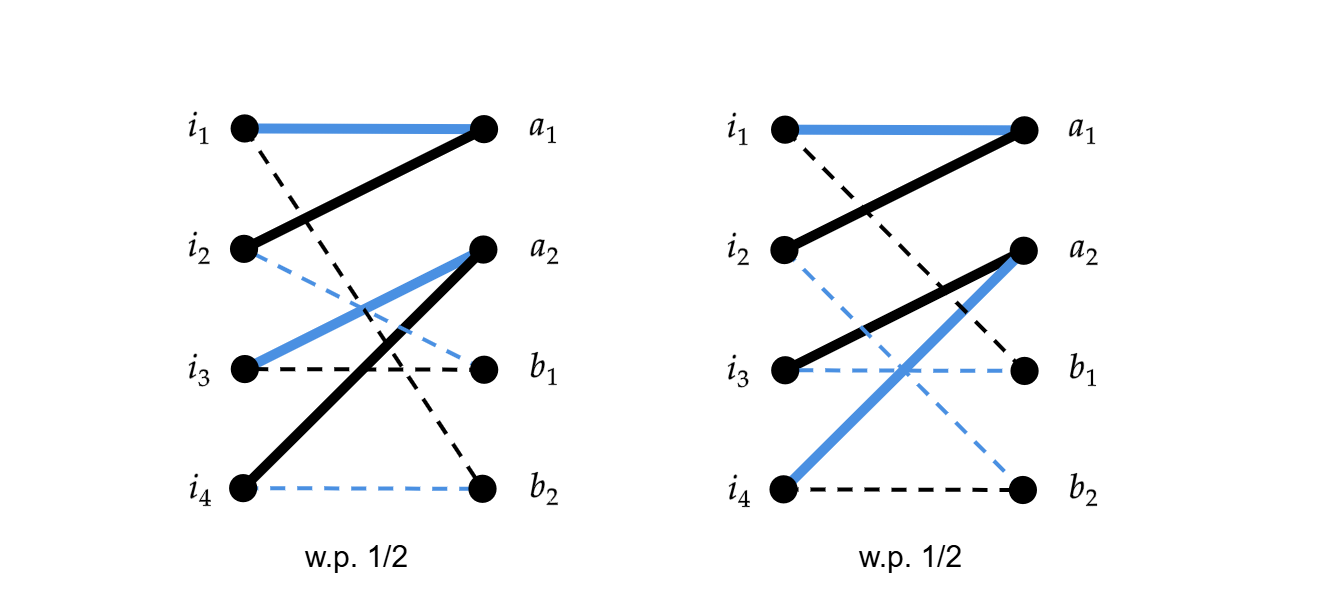}
\end{figure}

\Paragraph{Unweighted analysis}
We begin by showing that the vertex-weighted variant of \Cref{alg: round-augment} ($c = 1$) achieves an approximation ratio of $\nicefrac{7}{8}$ for the unweighted problem ($\tilde{w}_{e} = 1$), matching the LP integrality gap. We use $\nu(G)$ to denote the weight of the maximum matching in graph $G$; hence our task is to lower bound the expected max matching value $\E[\nu(G^\theta_A)]$ achieved in the second-stage graph $G^\theta=(B_2, I, E^\theta)$ after performing dependent rounding. Because the availabilities $(A_i)_{i \in I}$ are negatively associated, for the purposes of the analysis, a valid lower bound is given by assuming that the availability of offline nodes are fully independent after the first stage. We recall that we denote these independent random variables by $\tilde{A}$ using the notation introduced in \Cref{sec:prelims}. This is formalized in \Cref{lem:indsufficient}.

\begin{lemma} \label{lem:indsufficient}
    The expected gain of \Cref{alg: round-augment} is at least $\sum_{a \in B_1} \sum_i x_{ia} + \sum_{\theta \in \Theta} p^\theta \cdot \mathbb{E}[\nu(G^\theta_{\tilde{A}})]$, where $(\tilde{A}_i)_{i \in I}$ denotes a collection of independent Bernoulli random variables with $\Pr[\tilde{A}_i] = \Pr[A_i]$. 
\end{lemma}

\begin{proof}
    For each second-stage realization $\theta \in \Theta$ and subset of offline nodes $S \subseteq I$, let $f_{\theta}(S)$ denote the size of the maximum matching in the second-stage graph $G^\theta$ between node sets $B_2$ and $S$. Note that $f_{\theta} : 2^{I} \rightarrow \mathbb{Z}_{\ge 0}$ is a submodular function (it is the rank function of a transversal matroid). As NA random variables satisfy submodular dominance, we have 
    \begin{align*}
    \mathbb{E}[\text{ALG}] 
        &= \sum_{a \in B_1} \sum_{i \in I} x_{ia} + \sum_{\theta \in \Theta} p^{\theta} \cdot \mathbb{E}[f_{\theta}( \{ i : A_i = 1 \})] \\
        &\geq \sum_{a \in B_1} \sum_{i \in I} x_{ia} + \sum_{\theta \in \Theta} p^{\theta} \cdot \mathbb{E}[f_{\theta}( \{ i : \tilde{A}_i = 1 \})] \\
        &= \sum_{a \in B_1} \sum_{i \in I} x_{ia} + \sum_{\theta \in \Theta} p^\theta \cdot \mathbb{E}[\nu(G^\theta_{\tilde{A}})]. \qedhere
    \end{align*}
\end{proof}
When vertex availabilities are independent, we use the following lemma to lower bound the expected size of the maximum second-stage matching.

\begin{lemma} \label{lem:main}
    Let $G = (V, E)$ be a bipartite graph supporting a fractional matching $\vec{x} \in [0,1]^E$. If $(A_v)_{v \in V}$ are independent Bernoulli random variables satisfying $\Pr[A_v] \geq x_v := \sum_{u \sim v} x_{uv}$ for all $v \in V$, then \[\E[\nu(G_A)] \ge \sum_{e \in E} x_e - \frac{1}{2} \sum_{v \in V} x_v (1-\Pr[A_v]).\]
\end{lemma}

\begin{proof}
For convenience, let $p_v := \Pr[A_v = 1]$ for all $v \in V$. We proceed by induction on the number of edges $|E|$ in the graph, where for $|E|=0$ the statement trivially holds. The inductive step begins with three simple cases. 

If $G$ is not connected, we apply the inductive hypothesis to each connected component. If $G$ is connected but there exists some edge $e \in E$ with matching value $x_e = 0$, we apply the inductive hypothesis to the graph $G \setminus \{e\}$ with $\vec{x}$ restricted to $E \setminus \{e\}$; note that the right-hand-side is unchanged while the expected size of the maximum matching cannot increase. Finally, if there exists some edge $e=(u,v) \in E$ with matching value $x_{e} = 1$, note that $p_u = p_v = 1$. Thus, either $e$ is an isolated edge (and $G$ is disconnected) or there exists an adjacent edge $f$ with $x_f = 0$. In either case, we apply one of the two previous arguments. 

Next, we will show that it is sufficient to prove the result for acyclic graphs. Indeed, if $G$ contains a cycle $C$ in the support of $\vec x$, we can transform $\vec x$ to bring some $x_e \in (0,1)$ to $0$ or $1$. Because $C$ must be even length, partition $C=(e_1, \dots, e_{2k})$ into odd edges $C^+=\{e_1, e_3\, \dots\, e_{2k-1}\}$ and even edges $C^- = \{e_{2}, e_4, \dots, e_{2k}\}$. Consider setting
\[
\delta =\min  \left( \min_{e\in C^-}x_{e},\min_{e\in C^{+}}(1-x_{e}) \right),
\]
 and define a new fractional matching $\vec x'$ by
\[
x'_{e} =
\begin{cases}
x_{e}-\delta & e\in C^-,\\
x_{e}+\delta & e\in C^+,\\
x_{e}         & e\notin C.
\end{cases}
\]
Because $|C^+| = |C^-|$ and the fractional degrees of all nodes are preserved, the right-hand side is unchanged. Moreover, $x'_v \leq p_v$ for all $v \in V$. After this transformation, at least one edge in $C$ has $x_e' \in \{0,1\}$, and we apply one of the previous arguments.

It hence suffices to prove the statement for trees. We next show that it is only necessary to consider trees for which $p_v \neq x_v$ for at most one node in $G$. Toward that end, suppose that there exist two nodes $u, w \in V$ such that $x_u < p_u$ and $x_w < p_w$, and consider a path $P$ between $u$ and $w$. We will consider the cases where $|P|$ is odd or even separately.
\begin{enumerate}
    \item \textit{$|P|$ odd.} Partition $P = (e_1, \dots, e_{2k+1})$ into edge sets $P^+ = \{e_1, e_3, \dots, e_{2k+1}\}$ and $P^- = \{e_2, e_2, \dots, e_{2k}\}$. For $\delta = \min \left\{\min_{e \in P^-}x_e, \min_{e \in P^+} (1-x_e), p_u - x_u, p_w - x_w\right\},$ define a new fractional matching
    \[
        x'_{e} =
        \begin{cases}
        x_{e}-\delta & e\in P^-,\\
        x_{e}+\delta & e\in P^+,\\
        x_{e}         & e\notin P.
        \end{cases}
    \]
    Fractional degrees are unchanged for all nodes except $u$ and $w$, with $x_u' = x_u + \delta \leq p_u$ and $x_w' = x_w + \delta \leq p_w$. Proving the statement for $\vec{x}'$ immediately implies the statement for $\vec{x}$, as
    \begin{align*}
        \textup{RHS}(\vec{x}') - \textup{RHS}(\vec{x}) &=\left(\sum_{e \in E} x_e' - \frac{1}{2}\sum_{v \in V} x_v' \cdot (1-p_v)\right) - \left(\sum_{e \in E} x_e - \frac{1} {2}\sum_{v \in V} x_v \cdot (1-p_v)\right) \\ &= \delta  - \frac{1}{2}\left(\delta(1-p_u) + \delta(1-p_w)\right) \\
        &= \frac{\delta(p_u + p_w)}{2} \geq 0.
\end{align*}

\item \textit{$|P|$ even}. Assume without loss of generality that $p_u \geq p_w$ and the edge $e_1$ in $P = (e_1, \dots, e_{2k})$ is incident to node $u$. This time, partition $P$ into edge sets $P^+ = \{e_1,e_3, \dots, e_{2k-1}\}$ and $P^- = \{e_2, e_4, \dots, e_{2k}\}$. For $\delta = \min \left\{\min_{e \in P^-}x_e, \min_{e \in P^+} (1-x_e), p_u - x_u\right\}$, define a new fractional matching
\[
    x'_{e} =
    \begin{cases}
    x_{e}-\delta & e\in P^-,\\
    x_{e}+\delta & e\in P^+,\\
    x_{e}         & e\notin P.
    \end{cases}
\]
Once again, fractional degrees are unchanged except $x_u' = x_u + \delta \leq p_u$ and $x_w' = x_w - \delta \geq 0$. Moreover,
\begin{align*}
   \textup{RHS}(\vec{x}') - \textup{RHS}(\vec{x})
   &= \left(\sum_{e \in E} x_e' - \frac{1}{2}\sum_{v \in V} x_v' \cdot (1-p_v)\right) - \left(\sum_{e \in E} x_e - \frac{1}{2}\sum_{v \in V} x_v \cdot (1-p_v)\right) \\
   &= - \frac{1}{2}\left(\delta(1-p_u) - \delta(1-p_w)\right) \\
    &= \frac{\delta(p_u - p_w)}{2} \geq 0.
\end{align*}
\end{enumerate}
The result of the applicable transformation is that at least one of the following hold: (i) $x'_u = p_u$, $x'_w = p_w$, or (ii) $x'_e \in \{0,1\}$ for some edge $e$. In iteratively applying these transformations, we can either invoke the inductive hypothesis, or we terminate in a matching $\vec x$ with at most one $v \in V$ such that $x_v < p_v$.

Finally, we prove the statement for trees where $p_v = x_v$ for all $v \in V$ except possibly the root. For a non-root leaf $u$ with parent $w$, we define $G^\uparrow = G \setminus \{u\}$ and take $\vec{x}'$ to be the restriction of the fractional matching $\vec{x}$ to edges in $G^\uparrow$; note $x'_w = x_w - x_u$. Define new availabilities $A'_w =\textup{Ber}(p_w \cdot (1-p_u))$ and $A'_v = A_v$ for all $v \in V \setminus \{u,w\}$ for the graph $G^\uparrow$, and let $p'_v := \Pr[A'_v]$ for vertices $v \in V(G^\uparrow)$. Observe that $p'_w \ge x_w'$ because $$p_w \cdot (1 - p_u) = p_w \cdot (1 - x_u) \ge x_w \cdot (1 - x_u) \ge x_w - x_u. $$ 
We consider a matching strategy that matches $(u, w)$ if both $u$ and $w$ are active, and then computes a maximum matching in the remaining realized graph. The expected size of the matching returned by this method is a lower bound on $\E[\nu(G_A)]$; in particular, we have 
\begin{align*}
    \E[\nu(G_A)] &\geq p_u \cdot p_w + (1-p_u)p_w \cdot \E[\nu(G_A \setminus \{u\} \cup \{w\})] + \left(1- (1-p_u)p_w\right) \cdot \E[\nu(G_A \setminus \{u, w\})]\\
    &= p_u \cdot p_w + \E[\nu(G^\uparrow_{A'})].
\end{align*}
Applying the inductive hypothesis to $G^{\uparrow}$, $\vec{x}'$, and $A'$, we have
\begin{align*}
\E[\nu(G_A)] &\ge  p_u \cdot p_w + \sum_{e \in E} x_e - x_u -  \frac{1}{2} \left( \sum_{v \in V \setminus \{u\}}  x'_v \cdot (1-p'_v) \right) \\
&=  p_u \cdot p_w + \sum_{e \in E} x_e - x_u -  \frac{1}{2} \left( \sum_{v \in V}  x_v \cdot (1-p_v) - \ x_u \cdot (1-p_u) -  x_w \cdot (1-p_w) +   x_w' \cdot (1 - p_w') \right),
\end{align*} 
so it suffices to show $\Delta \ge 0$ for 
\begin{align*}
\Delta &:= p_u p_w - x_u
 + \tfrac12\bigl(x_u(1-p_u)+x_w(1-p_w)-x_w'(1-p_w')\bigr)  \\
 &=  x_u p_w - x_u
+\tfrac12\bigl(x_u(1-x_u)+x_w(1-p_w)-(x_w-x_u)\bigl(1-p_w(1-x_u)\bigr)\bigr).
\end{align*}
Observe
\[
\frac{\partial\Delta}{\partial p_w}
=x_u+\tfrac12\bigl[-x_w + (x_w-x_u)(1-x_u)\bigr]
= \frac{x_u}{2} \cdot (1 - x_w + x_u) \ge 0,
\]
hence it suffices to prove $\Delta \ge 0$ when $p_w=x_w$.  Finally, we calculate
\begin{align*}
2\Delta\bigl|_{p_w=x_w}
&= 2x_u x_w - 2x_u + x_u(1-x_u) + x_w(1-x_w) - (x_w - x_u)\bigl(1 - x_w(1-x_u)\bigr)\\ 
&= x_u(1-x_w)(x_w - x_u)\ge0
\end{align*}
where for the final inequality we used that $x_w \ge x_u$, as $u$ is a leaf adjacent to $w$. 
\end{proof}

\Cref{lem:main} implies an approximation ratio of $\nicefrac{7}{8}$ for \Cref{alg: round-augment}, run on unweighted graphs with scaling factor $c=1$. Indeed, we can lower bound $\E[\nu(G^\theta_{\tilde{A}})]$ by applying \Cref{lem:main} to the second-stage graph with indicators $\left((1)_{b\in B_2}, (\tilde{A}_i)_{i \in I}\right)$ and the fractional matching $y^{\theta}$. Constraint \eqref{eq:lp_online_constraint} of \eqref{LPon} ensures that $\Pr[\tilde{A}_i] = 1-x_i \ge \sum_{b \in B_2} y^{\theta}_{ib} = y_i^\theta$ for all $i \in I$. It follows that 
\begin{align}
\E[\nu(G^\theta_{\tilde{A}})] \ge \sum_{i \in I} y_i^{\theta} - \sum_{i \in I} \tfrac{1}{2} \cdot y_i^{\theta} \cdot (1 - (1-x_i)) = \sum_{i \in I}  y_i^{\theta}  \left(1 - \tfrac{1}{2} \cdot x_i \right) . \label{eq:78bound}
\end{align}
Hence we can bound our approximation ratio as 
\begin{align*}
    \frac{\mathbb{E}[\text{ALG}]}{\text{OPT\eqref{LPon}}} &\ge \frac{\sum_{a \in B_1} \sum_{i \in I} x_{ia} + \sum_{\theta} p^{\theta} \cdot \E[\nu(G^{\theta}_{\tilde{A}})]}{\sum_{a \in B_1} \sum_{i \in I} x_{ia} + \sum_{\theta} p^{\theta} \cdot \sum_{b \in B_2} \sum_{i \in I} y_{ib}^{\theta}} &&\text{(\Cref{lem:indsufficient})} \\
    &\ge \frac{\sum_{\theta} p^{\theta} \cdot \left(  \sum_{i \in I} x_i +  \sum_{i \in I} y_{i}^{\theta} \cdot \left( 1 - \tfrac{x_i}{2} \right) \right) }{\sum_{\theta} p^{\theta} \cdot \left( \sum_{i \in I} x_i + \sum_{i \in I} y_{i}^{\theta} \right)}. && \text{(\Cref{eq:78bound})}
\end{align*}
Because $x_i + y_i^{\theta} \le 1$, we conclude by bounding
\[ x_i + y_i^\theta\Bigl(1-\tfrac{x_i}{2}\Bigr) 
\ge x_i + y_i^{\theta} - \tfrac{1}{2} \left( \frac{x_i + y_i^{\theta}}{2} \right)^2 = (x_i + y_i^{\theta}) \cdot \left(1  - \tfrac{1}{2} \cdot  \frac{x_i + y_i^{\theta}}{4}  \right) \ge  \tfrac78\,(x_i+y_i^\theta)\] which demonstrates an approximation ratio of $\nicefrac{7}{8}$. 

\Paragraph{Vertex-weighted analysis}
The following lemma extends our unweighted result to the vertex-weighted setting (i.e., $\tilde{w}_{ij} = w_i$).
\begin{lemma} \label{lem:mainvtx}
    Consider a bipartite graph $G = (L, R, E)$ with non-negative vertex weights $(w_v)_{v \in R}$ that supports a fractional matching $\vec{x} \in [0,1]^E$. If $(A_v)_{v \in R}$ are independent Bernoulli random variables satisfying $\Pr[A_v] \geq x_v$ for all $v \in R$, then \[\E[\nu(G_A, w)] \ge  \frac{1}{2} \sum_{v \in R} w_v \cdot x_v (1 + \Pr[A_v]),\]
\end{lemma}
When all vertex weights are equal, this statement is an immediate corollary of \Cref{lem:main}. For the general case, we will consider the matching decisions made by the (optimal) greedy algorithm for transversal matroids. In particular, for a fixed bipartite graph $G=(L, R, E)$, we consider the transversal matroid $\mathcal{F}$ on subsets of $R$ where feasible sets are exactly those whose elements can be perfectly matched. Because the greedy algorithm is optimal for finding a max-weight feasible set in this matroid, we have the following useful result, which enables an inductive proof of \Cref{lem:mainvtx}.

\begin{fact}\label{fact:weighted-matching}
    Fix a bipartite graph $G=(L, R, E)$, fractional matching $x \in [0,1]^E$, and independent binary random variables $(A_v)_{v \in R}$. For each permutation $\sigma : [|R|] \rightarrow R$, there exists a maximum (vertex-weighted) matching algorithm $\mathcal{A}_{\sigma}$ such that for any vertex weights $w$ and $w'$ on $R$ with \[w_{\sigma(1)} \ge w_{\sigma(2)} \ge \ldots \ge w_{\sigma(n)} \quad \text{and} \quad w'_{\sigma(1)} \ge w'_{\sigma(2)} \ge \ldots \ge w'_{\sigma(n)},\] we have that\[\Pr[\mathcal{A}_\sigma \text{ matches } v \text{ in $(G_A, w)$}] = \Pr[\mathcal{A}_\sigma \text{ matches } v \text{ in $(G_A, w')$}] \quad \text{for all $v \in R$}.\]
\end{fact}

\begin{proof}[Proof of \Cref{lem:mainvtx}]
    Enumerate $R$ as $[n]$ and assume without loss of generality that $w_1 \geq w_2 \geq \dots \geq w_n$. We will prove the claim via reverse induction on the size of $S = \{i \in [n]: w_i = w_1\}$. For the base case $|S| = n$, we have from \Cref{lem:main} that
    \[\E[\nu(G_A, w)] \ge  w_1 \cdot \frac{1}{2} \cdot \sum_{v \in R}   x_v\cdot (1 + p_v). \]  
    Now, suppose that $w_1 =\dots = w_{k-1} > w_{k} \geq \dots \geq w_n$. For $\epsilon = w_{k-1} - w_k > 0$, consider new vertex weights $w'$ formed from $w$ by uniformly decreasing $w_1, \dots, w_{k-1}$ by $\epsilon$. Write
    \begin{align*}
    \Opt(w)   &= \mathbb{E}[\nu(G_A, w)] &\qquad
    \Opt(w')  &= \mathbb{E}[\nu(G_A, w')] \\
    \text{LB}(w)    &= \tfrac{1}{2} \sum_{i=1}^n w_i x_i (1 + p_i) &\qquad
    \text{LB}(w')   &= \tfrac{1}{2} \sum_{i=1}^n w_i' x_i (1 + p_i)
\end{align*}
    If $\mathcal{A}_{\textup{id}}$ is the algorithm guaranteed by \Cref{fact:weighted-matching} for the identity permutation $\sigma = \textup{id}$ and 
  $\pi_i := \Pr[\mathcal{A}_{\textup{id}}\text{ matches $i$ in $(G_A, w)$}]$, then 
  \[
    \Opt(w)=\sum_{i=1}^n w_i\,\pi_i \quad \text{and} \quad 
    \Opt(w')=\sum_{i=1}^n w_i'\,\pi_i, \]
since $w_1 \ge w_2 \ge \ldots \ge w_n$ and $w_1' \ge w_2' \ge \ldots \ge w_n'$. It follows that
    \[\Opt(w) -\Opt(w') = \epsilon\sum_{i=1}^{k-1} \pi_i \quad \text{and} \quad 
    \textup{LB}(w)-\textup{LP}(w')=\tfrac\epsilon2\sum_{i=1}^{k-1} x_i(1+p_i).
  \]
  Rearranging terms,
  \[
    \Opt(w)-\textup{LB}(w)
    = \biggl(\Opt(w')-\textup{LB}(w')\biggr)+\epsilon\left(\sum_{i=1}^{k-1}\pi_i-\tfrac12\sum_{i=1}^{k-1} x_i(1+p_i)\right).
  \]
  The inductive hypothesis applied to $w'$ implies $\Opt(w') \geq \text{LB}(w')$, so the first term on the right-hand-side is non-negative. For the second term, let $H$ be the node-induced subgraph of $G$ over nodes $L \cup \{1, \dots, k-1\}$. Applying \Cref{lem:main} to $H$ with unit node weights, availabilities $(A_i)_{i=1}^{k-1}$, and fractional matching $\vec{x}$ restricted to $H$ yields that $\sum_{i=1}^{k-1} \pi_i = \E[\nu(H_A)] \ge\tfrac12\sum_{i=1}^{k-1} x_i(1+p_i)$.  Thus $\Opt(w) \geq \textup{LB}(w)$, completing the proof.
\end{proof}

Applying \Cref{lem:mainvtx} to the second stage graph $G^\theta = (B_2, I, E^\theta)$ and matching $\vec{y}^\theta$ with independent node availabilities $(\tilde{A}_i)_{i \in I}$ implies a $\nicefrac{7}{8}$-approximation for the vertex-weighted setting in an analogous fashion to our earlier calculation. Indeed, 
\begin{align*}
    \frac{\mathbb{E}[\text{ALG}]}{\text{OPT\eqref{LPon}}} &\ge \frac{ \sum_{i \in I} x_i \cdot w_i + \sum_{\theta} p^{\theta} \cdot \E[\nu(G^{\theta}_{\tilde{A}}, w)]}{\sum_{i \in I} x_i \cdot w_i + \sum_{\theta} p^{\theta} \cdot \sum_{i \in I} y_{i}^{\theta} \cdot w_i} &&\text{(\Cref{lem:indsufficient})} \\
    &\ge \frac{\sum_{\theta} p^{\theta} \cdot \left( \sum_{i \in I} x_i \cdot w_i +  \tfrac{1}{2} \cdot \sum_{i \in I} y_i^{\theta} \cdot w_i \cdot (2-x_i) \right) }{\sum_{\theta} p^{\theta} \cdot \left( \sum_{i \in I} x_i \cdot w_i + \sum_{i \in I} y_i^{\theta} \cdot w_i \right)}.  && \text{(\Cref{lem:mainvtx})}
\end{align*}
We conclude by bounding \[x_i + \tfrac{1}{2} \cdot y_i^{\theta} \cdot (2 - x_i) = x_i + y_i^{\theta} - \frac{x_i \cdot y_i^{\theta}}{2} \ge x_i + y_i^{\theta} - \tfrac{1}{8} \cdot (x_i + y_i^{\theta})^2 \ge \tfrac{7}{8} \cdot (x_i + y_i^{\theta}).\]

\Paragraph{Tight vertex-uniform rounding of two-stage fractional matchings}

Up to this point, our approximation guarantees have been global rather than uniform for each vertex. We now show our vertex-weighted result implies that $\nicefrac{7}{8}$ is the tight vertex-uniform guarantee for rounding a fractional matching in two stages, hence proving the first half of \Cref{cor:fractionalmatching} (restated below for convenience). 

\fractionalmatchingbounds* 

To prove this result, we demonstrate the existence of a particular \emph{$\lambda$-bounded} contention resolution scheme (CRS) on transversal matroids. Given a vector $\lambda \in [0,1]^n$, this type of CRS guarantees that each element $i \in [n]$ is selected with probability at least $\lambda_i$. We note that uniform or ``$c$-balanced'' contention resolution schemes are a special case where $\lambda_i = c \cdot \Pr[i \in A]$.

\begin{definition}[$\lambda$-bounded CRS]
Let $\mathcal{F}\subseteq 2^{[n]}$ be a downwards‐closed family of feasible sets and let $A$ denote a random set of active elements from $[n]$. A \emph{$\lambda$‐bounded contention resolution scheme} for $(\mathcal{F},A)$ is a (randomized) mapping from realizations of $A$ to a subset $S \subseteq A$ with $S \in \mathcal{F}$
such that 
\begin{equation} \label{eq: bounded-crs}
    \Pr\bigl[i\in S\bigr]\ge \lambda_i
  \quad\forall\,i\in[n].
\end{equation}
\end{definition}

By LP duality (an idea found earlier in the literature, e.g. \cite{chekuri2014submodular}), a $\lambda$-balanced CRS exists when there is no ``proof'' of the infeasibility of \eqref{eq: bounded-crs} via a weighted subset of elements. For completeness, we formalize this claim in \Cref{lem:dualityCRS} and offer a short proof, noting that the idea was introduced in prior work.

\begin{lemma} \label{lem:dualityCRS}
There exists a $\lambda$-bounded contention resolution scheme for $(\mathcal{F},A)$ if and only if for every weight vector $\vec{w} \in \mathbb{R}_{\ge 0}^n$,
  \[
    \sum_{i=1}^n w_i \cdot \lambda_i \le 
    \E_{A} \Bigl[\max_{S \subseteq A, S \in \mathcal{F} }\sum_{i\in S}w_i\Bigr].
  \]
\end{lemma}

\begin{proof} The forward direction is immediate. For the backward direction, we encode the existence of a CRS admitting balance function $f(\cdot)$ via a system of linear inequalities, and then show that the implication follows from strong duality. Let $\Pr[A]$ denote the probability that set $A$ is active.  We introduce variables $p_{S \mid A}$ for all $A \subseteq [n]$ and $S \in 2^A \cap \mathcal{F}$, interpreted as the probability of selecting \(S\) with random active set \(A\). The following system of linear (in)equalities encodes the existence of the desired CRS. 

\begin{align}
 & \sum_{A\subseteq[n]}  \sum_{S \ni i}  p_{S\mid A}  \ge \lambda_i && \forall i \in [n]\\
& \sum_{S} p_{S\mid A} = \Pr[A] && \forall A \subseteq [n]\\
& p_{S\mid A}\ge0 && \forall A \subseteq [n], S \in 2^A \cap \mathcal{F}
\end{align}
Via duality this system is infeasible when there exists any solution to the following system:
\begin{align}
\sum_{i\in S} w_i + \beta_A &\le 0 && \forall A \subseteq[n], S \in 2^A \cap \mathcal{F} \\
\sum_{i=1}^n w_i \lambda_i - \sum_{A\subseteq [n]} \beta_A \Pr[A] &> 0 &&  \\
w_i &\ge 0 && \forall i \in [n]
\end{align}
Note that if there is a feasible solution to this system, there must be one with $$\beta_A = \min_{S \in 2^A \cap \mathcal{F}} \left(-\sum_{i\in S} w_i \right) = - \max_{S \in 2^A \cap \mathcal{F}} \sum_{i\in S} w_i.$$ With this substitution, the condition for primal infeasibility reduces to the existence of a weight vector $\vec{w} \in \mathbb{R}_{\ge 0}^n$ such that
\[
\sum_{i=1}^n w_i \lambda_i + \sum_{A\subseteq[n]} \left( \max_{S \in 2^A \cap \mathcal{F}} \sum_{i\in S} w_i \right) \Pr[A] > 0,
\]
which completes the proof of the lemma. 
\end{proof}

Returning to our result on rounding fractional matchings, it is straightforward to now demonstrate a $\nicefrac{7}{8}$ vertex-level rounding for two-stage fractional matching via our previous results. Given fractional matchings $\vec{x}$ and $\vec{y}$, we apply dependent rounding as in \Cref{alg: round-augment} in the first stage. By submodular dominance (as in \Cref{lem:indsufficient}) and \Cref{lem:mainvtx}, if $A$ denotes the availability status of offline nodes after dependent rounding, then
$$ \E[\nu(G_{A}, w)]  \ge \E[\nu(G_{\tilde{A}}, w)] \ge  \frac{1}{2} \cdot \sum_{i} w_i \cdot y_i  \cdot (2 - x_i)$$ for any weights $w$. 
By \Cref{lem:dualityCRS}, there exists a $\lambda$-bounded CRS to produce a matching for the second stage for $\lambda_i :=  \nicefrac{1}{2} \cdot y_i \cdot (2 - x_i).$ In the combined process of matching in the first stage according to dependent rounding and in the second according to this CRS, the total probability that any offline node $i \in I$ is matched is given by  $$\Pr[i \text{ matched}] \ge x_i + \frac{1}{2} \cdot y_i  \cdot (2 - x_i) \ge \frac{7}{8} \cdot (x_i + y_i),$$ where the final inequality holds as $x_i + y_i \le 1$. 

As a corollary, possibly of independent interest, we note that this result implies the existence of a CRS with uniformly \emph{additive} guarantees for transversal matroids. This is different from most CRS guarantees, which are uniform up to some \emph{multiplicative} constant. It is an interesting future direction to further generalize and apply this type of guarantee. 

\begin{corollary}
    For a transversal matroid feasibility constraint $\mathcal{F} \subseteq 2^{[n]}$, if $\vec{x} \in \textup{\textsf{conv}}(\mathcal{F})$ and every $i \in [n]$ is realized independently with probability $x_i$, there exists a contention resolution scheme with additive guarantee $$\Pr[i \in S] \ge \frac{1}{2} \cdot x_i \cdot (1 + x_i) \ge x_i - \frac{1}{8}.$$
\end{corollary}
\Section{Edge-weighted} \label{sec: edge-weight}

% \Paragraph{Integrality gap}
The edge-weighted variant of \eqref{LPon} has an integrality gap of at most $2\sqrt{2}-2$. We show this via the following example, which is related to the upper bound of \cite{naor2025online} for online rounding of fractional matchings. 

Consider an instance with $2n$ offline nodes. The first batch consists of $n$ online nodes $a = 1, 2, \ldots, n$ with node $a$ neighboring two offline nodes with both edges of weight 1, such that each offline node has exactly one neighbor in the first batch. In the second phase, a single online node is realized, neighboring a pair of offline nodes chosen uniformly at random, with both edges of weight $(1 + \sqrt{2}) \cdot n$. The optimal value of \eqref{LPon} is $(2 + \sqrt{2}) \cdot n$, achieved by setting the value of every variable to $\nicefrac{1}{2}$. The expected gain of an online algorithm that matches $k \cdot n$ online nodes in the first batch (for $k \in [0,1]$) is given by $k \cdot n + \left( 1 -  \frac{(kn)(kn-1)}{(2n)(2n-1)} \right) \cdot (1 + \sqrt{2})  \cdot n $. We can compute that 
\[\lim_{n \rightarrow \infty} \max_{k \in [0,1]} \frac{k \cdot n + \left( 1 -  \frac{(kn)(kn-1)}{(2n)(2n-1)} \right) \cdot (1 + \sqrt{2})  \cdot n}{(2 + \sqrt{2} ) \cdot n} = \max_{k \in [0,1]} \frac{ k + \left(1 - \frac{k^2}{4} \right) (1 + \sqrt{2}) }{2 + \sqrt{2}} = 2\sqrt{2} - 2.\]
This demonstrates the following claim. 

\begin{claim}
    There exists a sequence of instances for which the ratio of the expected value of optimum online and the optimal solution to \eqref{LPon} tends to $2\sqrt{2} - 2$. 
\end{claim}

\Paragraph{Edge-weighted analysis} We now show that the edge-weighted variant of \Cref{alg: round-augment} (with $c = 2\sqrt{2} - 2$) achieves an approximation ratio of $2\sqrt{2} - 2 \approx 0.828$ for the edge-weighted problem. This matches the integrality gap of the edge-weighted LP.

Our approach is to lower bound the weight of the maximum matching in the second stage by the size of a matching produced by an algorithm which recursively matches within a max-depth star subgraph, then proceeds on the remaining available graph. Note that the desired approximation would be immediate if this procedure could guarantee that each edge $e \in E$ is included in the final matching with probability at least its scaled matching weight $c \cdot y_e$. \Cref{lemma: star-CRS} shows constructively that such a guarantee can be achieved for star graphs (i.e., 1-uniform matroids) by a monotone contention resolution scheme. We will inductively apply this guarantee to star graphs pulled from a tree, using monotonicity in the CRS to argue that Negative Association is preserved. The proof relies in part on the algebraic bound in Fact~\ref{fact:star-bound}, a proof of which can be found in \Cref{app:deferredproofs}.

\begin{restatable}{fact}{factCRSineq} \label{fact:star-bound}
    For $c:=2 \sqrt{2} - 2$ and $x \geq 3$,
    \[c^{x+1}\left(1-\frac{1}{x+1}\right)^{x+1} \leq c^{x}\left(1-\frac{1}{x}\right)^{x}.\]
\end{restatable}

\begin{lemma}[Star CRS]\label{lemma: star-CRS}
Fix $c := 2\sqrt{2} - 2$. Let $\{y_i\}_{i=1}^n$ be non-negative numbers with $\sum_i y_i \le 1$ and $(A_i)_{i=1}^n$ be NCD Bernoulli random variables satisfying $\Pr[A_i] \ge 1 - c \cdot (1 - y_i)$ for all $i \in [n]$. Say index $i$ is \emph{active} if $A_i=1$. There exists a monotone\footnote{By \emph{monotone} we refer to the property that for $i \in A \subseteq B$, $\Pr[i \text{ selected} \mid A \text{ available}] \ge \Pr[i \text{ selected} \mid B \text{ available}]$.}contention resolution scheme, which upon observing the realizations $(A_i)_{i=1}^n$, selects at most one active index such that $$\Pr[i \textup{ selected}] = c \cdot y_i.$$
\end{lemma}

\begin{proof}
Existence of such a selection scheme is equivalent to the condition $\sum_{i \in K} c \cdot y_i \le \Pr[\bigcup_{i \in K} A_i]$ for all subsets $K \subseteq \{1, \dots, n\}$. This is analogous to \Cref{lem:dualityCRS}, with the added requirement of monotonicity; the proof follows ideas of \cite{chekuri2014submodular}, and for completeness can be found in \Cref{app:deferredproofs}. Using the negative cylinder dependence of $(A_i)_{i=1}^n$ and the bounds on individual marginal probabilities, we have $$\Pr \left[ \bigcup_{i \in K} A_i \right] = 1 - \Pr \left[ \bigwedge_{i \in K} \overline{A_i} \right] \ge 1 - \prod_{i \in K} \Pr [ \overline{A_i} ] \ge 1 - \prod_{i \in K} c(1-y_i). $$
Thus, it suffices to show
\begin{equation}
c \sum_{i \in K} y_i + c^{|K|} \prod_{i \in K} (1-y_i) \le 1. \label{ineq:onelevelCRS_NCD_final}
\end{equation}  
For convenience, define $S_K := \sum_{i \in K} y_i$ and $m := |K|$. Note that $\prod_{i \in K} (1-y_i) \le (1-S_K/m)^{m}$ by the AM-GM inequality. So, to prove \eqref{ineq:onelevelCRS_NCD_final}, it is sufficient to show that
\[f(x,m) := c x + c^m (1-x/m)^m \le 1\]
for all $x \in [0,1]$ and $m \geq 1$. The partial derivative of $f$ with respect to $x$ is
\begin{align*}
    \frac{\partial f}{\partial x} 
        = c - c^m (1-x/m)^{m-1} 
        \geq c - c^m 
        \geq 0.
\end{align*}
 Thus, $f(x, m)$ achieves its maximum at $x=1$ for any fixed $m \geq 1$. Define $h(m) := F(1,m)$. We need to show $h(m) = c + c^m(1-1/m)^m\le 1$. It suffices to consider $m \in \{1, 2, 3\}$ since $h(m+1) \leq h(m)$ for all $m \geq 3$ by \cref{fact:star-bound}. Computing each case explicitly, $h(1) = c \le 1$. $h(2) = c + c^2(1/2)^2 = c + c^2/4$. Since $c = 2\sqrt{2}-2$ is a positive root of $x+x^2/4=1$ (i.e., $x^2+4x-4=0$), we have $h(2)=1$. Finally, $h(3) = c + c^3(2/3)^3 \approx 0.996$.
\end{proof}

To extend our result to general graphs, we apply an inductive argument where we match according to the monotone CRS guaranteed by \Cref{lemma: star-CRS} within a max-depth star with root-most node $g$, then recursively compute a matching in the remaining graph with an updated availability $A_g'$ for node $g$. To couple with the event that $g$ remains available for matching, we select $A_g'$ such that $\Pr[A_g']$ is exactly the probability that $g$ is initially available and also not matched in the star.

Because the existence of this CRS depends on the NA property of node availabilities, it is critical that availabilities remain NA throughout the induction. As $g$'s new availability may have correlations with the availabilities of other nodes we require a more careful argument to handle this step compared to our previous vertex-weighted analysis which had independence. To do use, we use the monotonicity of our star-CRS along with the closure properties of NA (see \Cref{sec:prelims}), along with the following basic fact showing that applying independent randomness to one random variable in a collection of NA random variables maintains the NA property (see \Cref{app:deferredproofs} for a proof). % This is exactly how we form $A_g'$ in the proof of \Cref{lemma: main-edge}.
\begin{restatable}{claim}{claimNAprelim}\label{claim:NAprelim}
    If \((X_0, X_1, \ldots, X_n)\) are NA random variables with \(X_0 \in [0,1]\) almost surely, then the random variables \((\textup{Ber}(X_0), X_1, \ldots, X_n)\) are also NA, where $\textup{Ber}(X_0)$ is realized according to independent randomness. 
\end{restatable}

\begin{lemma}[Main Edge-Weighted Lemma] \label{lemma: main-edge}
 Fix $c:= 2\sqrt{2}-2$. Let $G=(L, R, E)$ be a bipartite graph with non-negative edge weights $\{w_e\}_{e \in E}$ supporting a fractional matching $\vec{x}$. If $(A_v)_{v \in R}$ are NA binary random variables satisfying $\Pr[A_v] \ge 1 - c\cdot (1-x_v)$ for all $v \in R$, then
\[\mathbb{E}[\nu(G_A, w)] \ge c \cdot \sum_{e \in E} w_e x_e.\]
\end{lemma}

\begin{proof}
We proceed by induction on the size of the edge set $|E|$. The statement trivially holds for $|E|=0$. We begin the inductive step with the same simple cases as in the vertex-weighted setting. If $G$ is disconnected, we apply the inductive hypothesis to each connected component. If $G$ is connected but there exists some edge $e \in E$ with $x_e = 0$, we apply the inductive hypothesis to the graph $G \setminus \{e\}$ with $\vec{x}$ restricted to $E \setminus \{e\}$; under this transformation, the right-hand-side is unchanged while the weight of the maximum matching cannot increase. Finally, if there exists an edge $e \in E$ with $x_{e} = 1$, either $e$ is isolated (and $G$ is disconnected) or there exists an adjacent edge with zero matching value. In either case, we apply one of the two previous arguments.

Cycles in the support of $\vec{x}$ can be eliminated in a manner similar to the proof of \Cref{lem:main}. The only modification for the edge-weighted setting is that between the edge sets $C^+$ and $C^-$, matching values should be increased on the set with higher total edge weight. This operation cannot decrease the right-hand side of the objective, and as before, has no impact on the size of the max matching in $G_A$.

For the remainder of the proof, we assume that the support of $\vec{x}$ is a tree. Thus, it suffices to prove the statement when $G$ is a tree. We root it arbitrarily and consider two cases.

\begin{enumerate}
    \item \textit{There exists a (deterministic) leaf node $u \in L$.} Let $v \in R$ be its unique neighbor and define an independent Bernoulli random variable $Z := \text{Ber} \left( \nicefrac{c \cdot x_{uv}}{1 - c\cdot (1-x_{v})} \right).$ In this case, we analyze a matching strategy for $G_A$ that selects the edge $e=(u,v)$ if both $A_{v} = 1$ and $Z = 1$, deletes $e$ otherwise, and then augments with a max matching in the remaining graph. Toward that end, we define a new collection of random variables $A'$ where $A'_v = A_v \cdot (1 - Z)$ and $A'_{v'} = A_{v'}$ for all $v'\neq v \in R$. Note that $A'$ are matching availabilities in the remaining graph $G^\uparrow = G \setminus \{u\}$ under the matching strategy described above. Thus,
    \[\E[\nu(G_A, w)] \ge c \cdot x_{uv} \cdot w_{uv} + \E[\nu(G^\uparrow_{A'}, w)].\]
    It is immediate from the independence of $Z$ that the collection of random variables $A'$ is NA. Moreover, if $\vec{x}'$ denotes the restriction of $\vec{x}$ to $G^\uparrow$, then $x'_v = x_v - x_{uv}$. So,
    \[\Pr[A_v'] = \Pr[A_{v}] \cdot (1 - \Pr[Z]) \ge 1- c \cdot(1- x_{v}) \cdot \left( 1 - \frac{c \cdot x_{uv}}{1 - c\cdot (1-x_{v})} \right) = 1 - c \cdot(1- x_{v}').\]
    We conclude the proof of this case by applying the inductive hypothesis to bound
    \[\E[\nu(G^\uparrow_{A'}, w)] \geq c \cdot \sum_{e \in E \setminus (u,v)} w_e \cdot x_e.\]

    \item  \textit{No leaf is in $L$.} Given that $G$ is a non-empty tree, let $S \subseteq R$ be a maximal set of leaf nodes with shared parent $p \in L$. If $p$ is the root node, then $G$ is a star with center $p$ and leaves $S$. Let $\phi$ be the CRS guaranteed by \Cref{lemma: star-CRS} for the random variables $(A_i)_{i \in S}$, where $\phi\left((A_i)_{i \in S}\right)$ denotes the (random) node selected from among the active elements of $(A_i)_{i \in S}$. Consider a matching strategy that matches $p$ to the node $\phi\left((A_i)_{i \in S}\right)$; if $p$ is the root this yields a random matching with expected weight
    \[\E\left[\sum_{i \in S} w_{ip} \cdot \mathbbm{1}\{\phi\left((A_i)_{i \in S}\right) = i\}\right] = c \cdot \sum_{i \in S} w_{ip} \cdot x_i = c \sum_{e \in E} w_e \cdot x_e.\]
    
    Now, suppose instead that $p$ has some parent $g \in R$, and define the node-induced subgraphs $G^{\downarrow} := G[S \cup \{p, g\}]$ and $G^{\uparrow} := G[V \setminus S \setminus \{p\}]$. Let $\phi$ denote the CRS from \Cref{lemma: star-CRS} applied to the random variables $(A_i)_{i \in S}, A_g$. We will consider a matching strategy that matches $p$ to the node $\phi((A_i)_{i \in S}, A_g)$ selected by the CRS on $G^{\downarrow}_A$, then augments with a max matching in the remaining graph. Given $\phi$, define a function $f$ such that 
    \begin{align*}
       f((A_i)_{i \in S}, A_g ) 
        &= A_g \cdot (1 - \Pr[g = \phi((A_i)_{i \in S}, A_g) \mid (A_i)_{i \in S}, A_g]).
    \end{align*}
    Note $f((A_i)_{i \in S}, A_g)$ is a random variable that takes value 0 if $A_g = 0$ and otherwise gives the conditional probability that $\phi$ does \emph{not} select $g$ (i.e., does not match along edge $(p,g)$). For $v \in V \setminus S \setminus \{p, g\}$, define $A'_v = A_v$; define also $A_g' = \text{Ber}(f((A_i)_{i \in S}, A_g ))$ where this Bernoulli is realized according to independent randomness. Note that $A'$ can be perfectly coupled to the availabilities of nodes in $G^{\uparrow}$ after applying our CRS-based matching to $G^{\downarrow}$: in particular, if $A_g = 0$ we have $A'_g = 0$, while if $A_g = 1$, the availability $A'_g$ for a fixed input $(A_i)_{i \in S}, A_g$ is determined by an independent sample with probability $f((A_i)_{i \in S}, A_g)$. Thus, we lower bound the expected max matching weight in $G_A$ by following our CRS on $G^{\downarrow}$ and augmenting with the maximum matching among the remaining available nodes in $G^{\uparrow}$. \begin{align*}
        \E[\nu(G_A, w)] &\geq \sum_{e \in G^\downarrow} w_e \cdot \Pr[e \text{ selected by $\phi$}] + \E[\nu(G^\uparrow_{A'}, w)].
    \end{align*}

    Before we can apply the inductive hypothesis, we need to (1) show that the collection of random variables $A'$ is NA, and (2) construct a feasible matching $\vec{x}'$ such that $\Pr[A_v'] \geq  1-c\cdot (1-x_v')$ for all $v \in R \setminus\{S\}$.
    For (1), note that the function $f$ is clearly increasing in $A_g$. Moreover, because the CRS $\phi$ is monotone, $f$ is also increasing in each $A_i$ for $i \in S$. Thus, by the closure properties of NA \cite{joag1983negative}, we have that $f((A_i)_{i \in S}, A_g), (A_v)_{v \in R \setminus S \setminus \{g\}}$ are NA. By \Cref{claim:NAprelim}, this means $\text{Ber}(f((A_i)_{i \in S}, A_g)), (A_v)_{v \in R \setminus S \setminus \{g\}}$ are also NA. For (2), we simply take $\vec{x}'$ to be the restriction of $\vec{x}$ to $G^\uparrow$. This is clearly a feasible fractional matching in $G^\uparrow$. Moreover, $\Pr[A'_v] \geq c\cdot x_v'$ holds for all $v \in R \setminus S \setminus \{g\}$ by assumption since $x_v'=x_v$ and $A_v'=A_v$. At $g$, because $x_g' = x_g - x_{pg}$,
    \begin{align*}
        \Pr[A_g'] 
            &= \E[f((A_i)_{i \in S}, A_g)] \\
            &= \E[A_g] - \Pr[g = \phi((A_i)_{i \in S}, A_g)]\\
            &\geq 1 - c \cdot (1-x_g) - c \cdot x_{pg} \\
            &= 1 - c \cdot (1- x_g').
    \end{align*}
    Finally, applying the inductive hypothesis to lower bound $\E[\nu(G^\uparrow_{A'}, w)]$ yields the desired result:
    \begin{align*}
        \E[\nu(G_A, w)] &\geq \sum_{e \in G^\downarrow} w_e \cdot \Pr[e \text{ selected by $\phi$}] + c\sum_{e \in G^\uparrow} w_e x'_e \\
        &= c \sum_{e \in G_\downarrow} w_ex_e + c \sum_{e \in G^\uparrow}w_ex_e \\
        &= c \sum_{e \in E} w_e x_e. \qedhere
    \end{align*}
\end{enumerate} \end{proof}
Applying \Cref{lemma: main-edge} to the second-stage graph $G^\theta$, matching $\vec{y}^\theta$, and NA node availabilities $(A_i)_{i \in I}$ immediately implies a $(2\sqrt{2}-2)$-approximation for the edge-weighted setting. Indeed, 
\begin{align*}
    \frac{\mathbb{E}[\text{ALG}]}{\text{OPT\eqref{LPon}}} &\ge \frac{ \sum_{e \in E} c \cdot x_e \cdot w_e + \sum_{\theta} p^{\theta} \cdot \E[\nu(G^{\theta}_{A}, w)]}{\sum_{e \in E} x_e \cdot w_e + \sum_{\theta} p^{\theta} \cdot \sum_{e \in E^\theta} y_{e}^{\theta} \cdot w_e} && \\
    &\ge \frac{\sum_{\theta} p^{\theta} \cdot \left(\sum_{e \in E} c \cdot x_e \cdot w_e + \sum_{e \in E^{\theta}} c \cdot y^\theta_e \cdot w_e \right) }{\sum_{\theta} p^{\theta} \cdot \left( \sum_{e \in E} x_e \cdot w_e + \sum_{e \in E^\theta} y_e^{\theta} \cdot w_e \right)}  = c.   && \text{(\Cref{lemma: main-edge})}
\end{align*}

\Paragraph{Tight fractional rounding} The bound above implies an edge-uniform rounding of fractional matchings up to the multiplicative constant $c_{\textsf{E}} := 2\sqrt{2}-2$. Indeed, our dependent rounding algorithm applied to first-stage fractional matching $\vec{x}$ matches each first-stage $e$ with probability exactly $c \cdot x_e$. Furthermore, combining \Cref{lem:dualityCRS} with \Cref{lemma: main-edge} implies that in the second stage, there exists a contention resolution scheme matching each edge $e$ in the support of a feasible fractional matching $\vec{y}$ with probability $c \cdot y_e$. This proves the edge-weighted half of \Cref{cor:fractionalmatching}.\footnote{Naturally, we apply \Cref{lem:dualityCRS} to a ground set of edges, realized according to some correlated distribution.}
\Section{Two-Stage Matching from Samples}

The algorithms we have provided thus far assume full knowledge of the distribution over second-stage graphs $\mathcal{D}$, and run in time that is polynomial in the size of the support of $\mathcal{D}$. Practically, both of these dependencies can be undesirable. In some settings, an exact representation of $\mathcal{D}$ may be unknown, and in others,  $\mathcal{D}$ may have a compact representation despite having large support. Consider, for example, the setting originally studied by \cite{feng2021two} in which each vertex $b \in B_2$ appears independently with known probability $\pi_b$. This distribution has exponential support but is succinctly represented by a vector of probabilities.

To address both of these points, we extend our results to a setting where we only have access to i.i.d. samples from $\mathcal{D}$. We differentiate between algorithms that have access to the full distribution $\mathcal{D}$ and those that only have access to the empirical distribution $\hat{\mathcal{D}}$ over samples using superscripts. In the vertex-weighted setting, we show that running \textsf{Round-Augment} on an empirical distribution $\hat{\mathcal{D}}$ over $\text{poly}(|I|, \epsilon^{-1})$ samples suffices to guarantee a $(\nicefrac{7}{8}-\epsilon)$-approximation  (\Cref{thm: vtx-samples}).

\begin{theorem} \label{thm: vtx-samples}
For a distribution $\mathcal{D}$ over vertex-weighted graphs, let $\hat{\mathcal{D}}$ be the empirical distribution over samples $\theta_1, \dots, \theta_k \sim \mathcal{D}$. With probability at least $1-\delta$ over the draw of the $k \geq O\left(\frac{|I| + \log(\nicefrac{1}{\delta})}{\epsilon^2} \right)$ samples,
\[\underset{\theta \sim \mathcal{D}}{\E}[\textup{ALG}^{\hat{\mathcal{D}}}(\theta)] \geq \left(\frac{7}{8} - \epsilon\right) \cdot \underset{\theta \sim \mathcal{D}}{\E}[\OptOn^\mathcal{D}(\theta)].\]
\end{theorem}
\begin{proof}
    For any subset $S \subseteq I$ of offline nodes matched in the first stage, let $S(\theta)$ denote the offline nodes matched in a second-stage max matching under realization $\theta \in \Theta$. Define the expected vertex-weighted matching values $\mu_S = \sum_{v \in S} w_i + \E_{\theta \sim \mathcal{D}}[\sum_{v \in S(\theta)} w_v]$ and $\hat{\mu}_S = \sum_{v \in S} w_v + \E_{\theta \sim \hat{\mathcal{D}}}[\sum_{v \in S(\theta)} w_v]$ under the full and empirical distributions, respectively. We assume without loss of generality that the first-stage matchings generated by $\OptOn^\mathcal{D}$ and $\OptOn^{\hat{\mathcal{D}}}$ are integral---it will be convenient to refer to the first-stage node set matched by $\OptOn^\mathcal{D}$ as $\Opt$ and the first-stage node set matched by $\OptOn^{\hat{\mathcal{D}}}$ as $\hat{\Opt}$.
    
    To begin, we will show that $k$ samples is sufficient to approximate the difference $\mu_{S} - \mu_{T}$ to high accuracy for any pair of node sets $S, T \subseteq I$. Note that $\hat{\mu}_{S} - \hat{\mu}_{T}$ is the sample average of bounded, i.i.d. random variables $D(\theta_1), \dots, D(\theta_k)$ where $D(\theta_i) := \sum_{v \in S \cup S(\theta_i)} w_v - \sum_{v \in T \cup T(\theta_i)} w_v$. Indeed,
    \[\sum_{v \in T \cup T(\theta)} w_v \geq \sum_{v \in T \cup S(\theta)} w_v \geq \sum_{v \in S \cup S(\theta)} w_v - \sum_{v \in S \setminus T} w_v. \]
    The observation that $S \setminus T$ is a first-stage matching candidate for $\Opt$ implies
    $D(\theta_i) \leq \sum_{v \in S \setminus T} w_v \leq \mu_{\Opt}.$ Thus, without loss of generality $D(\theta_i) \in [-\mu_\Opt, \mu_\Opt]$. Applying a Hoeffding bound with $k \geq \frac{2}{\epsilon^2}\left(2\ln2 \cdot |I| + \ln(\nicefrac{2}{\delta})\right)$ samples, we have
    \[\Pr\left[\left|(\hat{\mu}_S - \hat{\mu}_T) - (\mu_S - \mu_T)\right| \geq \epsilon \cdot \mu_{\Opt}\right] \leq 2 \cdot \exp(\nicefrac{-k\epsilon^2}{2}) \leq \frac{\delta}{2^{2|I|}}.\]
    
    A union bound over the $2^{2|I|}$ possible pairs $S, T \subseteq I$ gives that the approximation holds for all pairs of sets simultaneously with probability at least $1 - \delta$. In particular, this holds for arbitrary $S$ when $T = \Opt$. Because $\textup{ALG}^{\hat{\mathcal{D}}}$ is a potentially random strategy that mixes between sets of first-stage matched nodes,
    \begin{equation} \label{eq: vtx-hoeffding}
        \underset{\theta \sim \mathcal{D}}{\E}[\textup{ALG}^{\hat{\mathcal{D}}}(\theta)] - \underset{\theta \sim \hat{\mathcal{D}}}{\E}[\textup{ALG}^{\hat{\mathcal{D}}}(\theta)] \geq \frac{7}{8}\left(\mu_{\Opt} - \hat{\mu}_{\Opt}\right) - \epsilon \cdot \mu_{\Opt}
    \end{equation}
    It follows from \Cref{lem:mainvtx} and the optimality of $\hat{\Opt}$ over the empirical distribution that
    \begin{equation} \label{eq: empirical-bound}
        \underset{\theta \sim \hat{\mathcal{D}}}{\E}[\textup{ALG}^{\hat{\mathcal{D}}}(\theta)] \geq \frac{7}{8} \cdot \hat{\mu}_{\hat{\Opt}} \geq \frac{7}{8} \cdot \hat{\mu}_{\Opt}.
    \end{equation}
    Adding the bounds from \eqref{eq: vtx-hoeffding} and \eqref{eq: empirical-bound} yields the result.
\end{proof}

To achieve a similar $\epsilon$ loss in the approximation ratio for the edge-weighted setting, there is a necessary dependence on the scale of edge weights and the expected matching weight of the optimum online algorithm. For instances with max matching weight $W$ achievable across any realization and expected online optimum matching weight $\mu$, we show that running \textsf{Round-Augment} on an empirical distribution over $\text{poly}(\frac{W}{\mu}, |E|, \epsilon^{-1})$ samples suffices to guarantee an approximation at least $(2\sqrt{2} - 2)-\epsilon$ (\Cref{thm:edgesamples}). 
\begin{theorem} \label{thm:edgesamples}
For a distribution $\mathcal{D}$ over edge-weighted graphs, let $\hat{\mathcal{D}}$ be the empirical distribution over samples $\theta_1, \dots, \theta_k \sim \mathcal{D}$. If $W$ is the sum of edge weights and $\mu$ is the expected matching weight of the optimum online, then with probability at least $1-\delta$ over the draw of the $k \geq O\left(\frac{W^2}{\mu^2 \epsilon^2} \left(|E| + \log(\nicefrac{1}{\delta})\right)\right)$ samples,
    \[\underset{\theta \sim \mathcal{D}}{\E}[\textup{ALG}^{\hat{\mathcal{D}}}(\theta)] \geq \left(2\sqrt{2} - 2 - \epsilon\right) \cdot \underset{\theta \sim \mathcal{D}}{\E}[\OptOn^\mathcal{D}(\theta)].\]
\end{theorem}
\begin{proof}
    We use the same general ideas as in the proof of \Cref{thm: vtx-samples}. For any first-stage matching $M \subseteq E$, let $M(\theta)$ denote the second-stage max matching corresponding to $\nu(G^\theta \setminus M, w)$ for realization $\theta \in \Theta$.  Define the expected edge-weighted matching values $\mu_M = \sum_{e \in M} w_e + \E_{\theta \sim \mathcal{D}}[\nu(G^\theta \setminus M)]$ and $\hat{\mu}_M = \sum_{e \in M} w_e + \E_{\theta \sim \hat{\mathcal{D}}}[\nu(G^\theta \setminus M)]$ under the full and empirical distributions, respectively. Once again, we assume without loss of generality that $\OptOn^\mathcal{D}$ and $\OptOn^{\hat{\mathcal{D}}}$ produce integral first-stage matchings $\Opt$ and $\hat{\Opt}$, respectively.

    This time, we show that $k$ samples is sufficient to approximate the difference $\mu_{M_1} - \mu_{M_2}$ to high accuracy for any pair of first-stage matchings $M_1$ and $M_2$. Note that $\hat{\mu}_{M_1} - \hat{\mu}_{M_2}$ is the sample average of i.i.d. random variables taking values in $[-W, W]$. By Hoeffding's inequality with $k\geq \frac{2W^2}{\mu^2\epsilon^2}\left[2\ln2 \cdot |E| + \ln(\nicefrac{2}{\delta})\right]$ samples,
     \[\Pr\left[\biggl|(\hat{\mu}_{M_1} - \hat{\mu}_{M_2}) - (\mu_{M_1} - \mu_{M_2}) \biggr| \geq \epsilon \cdot \mu_{\Opt}\right] \le 2 \cdot \exp\left(-\frac{k\epsilon^2\mu_{\Opt}^2}{2W^2}\right) \leq \frac{\delta}{2^{2|E|}}.\]
     for each pair of first-stage matchings $M_1, M_2$. A union bound over the $2^{2|E|}$ such pairs yields that
      \begin{equation} \label{eq: edge-hoeffding}
        \underset{\theta \sim \mathcal{D}}{\E}[\textup{ALG}^{\hat{\mathcal{D}}}(\theta)] - \underset{\theta \sim \hat{\mathcal{D}}}{\E}[\textup{ALG}^{\hat{\mathcal{D}}}(\theta)] \geq (2\sqrt{2}-2) \cdot \left(\mu_{\Opt} - \hat{\mu}_{\Opt}\right) - \epsilon \cdot \mu_{\Opt}
    \end{equation}
     It follows from \Cref{lemma: main-edge} and the optimality of $\hat{\Opt}$ over the empirical distribution that
    \begin{equation} \label{eq: empirical-bound-edge}
        \underset{\theta \sim \hat{\mathcal{D}}}{\E}[\textup{ALG}^{\hat{\mathcal{D}}}(\theta)] \geq (2\sqrt{2}-2) \cdot \hat{\mu}_{\hat{\Opt}} \geq (2\sqrt{2}-2) \cdot \hat{\mu}_{\Opt}.
    \end{equation}
    Adding the bounds from \eqref{eq: edge-hoeffding} and \eqref{eq: empirical-bound-edge} yields the result.
\end{proof}
\Section{Discussion} \label{sec:discussion}
In this work, we proposed the \textsf{Round-Augment} framework for stochastic two-stage bipartite matching and used this framework to derive algorithms for the vertex-weighted and edge-weighted settings that achieve approximation ratios of $\nicefrac{7}{8}$ and $2\sqrt{2} - 2 \approx 0.828$ to the online optimum, respectively. Furthermore, we showed that both approximation ratios are tight, in the sense that they match the integrality gap of our linear relaxations of the optimum online policy.

Our work motivates an interesting new direction in the burgeoning literature on approximating the optimum online policy (i.e., philosopher inequalities). We provided a tight rounding of the natural ``marginal-level'' LP relaxation of the optimum online. Despite a line of work on the online bipartite matching problem with node-arrivals, a guarantee matching the integrality gap is currently not known for the natural ``marginal level'' LP relaxation of optimum online. Is it possible to obtain one? Further improvements to the two-stage problem will require a new LP relaxation, perhaps utilizing a hierarchy to capture higher moments, which is of interest to a number of problems in the space of approximating optimum online. 
\section*{Acknowledgments}
This work was supported in part by AFSOR grant FA9550-23-1-0251, ONR grant N000142212771, and a National Defense Science \& Engineering Graduate (NDSEG) fellowship. We thank Ellen Vitercik for early discussions that helped lead us to this research direction.

\bibliographystyle{alpha}
\bibliography{abb,a,ultimate}

\appendix

\Section{Hardness} \label{app:hardness}

A result of \cite{braverman2025new} can be adapted to show that the two-stage bipartite matching problem is NP-hard to approximate. For the sake of completeness, we provide the construction below. 

\begin{definition}
    An instance $\phi$ of the \textsc{Max-(3,4)-SAT} problem consists of a 3-\textsc{SAT} instance where each of $n$ variables appears in exactly 4 of the $m$ clauses. 
\end{definition}

We reduce from a hardness of approximation result for this problem. 

\begin{theorem}[\cite{berman2003approximation}] \label{thm:max34}
    It is NP-hard to approximate \textsc{Max-(3,4)-SAT} within some universal constant $\alpha < 1$. 
\end{theorem}

We can use this to demonstrate the two-stage stochastic matching problem is NP-hard to approximate within a constant, even in its simplest unweighted variant.  

\begin{lemma}
    There exists some constant $\beta < 1$ such that the unweighted two-stage stochastic matching problem is hard to approximate within factor $\beta$. 
\end{lemma}

As the proof proceeds nearly identically to \cite{braverman2025new}, we provide a sketch below. 

\begin{proof}[Proof Sketch.]
    Given a \textsc{Max-(3,4)-SAT} instance $\phi$ with $n$ variables and $m$ clauses, construct a two-stage graph $G_{\phi}$ with $2n$ offline nodes $\{\texttt{T}_1, \texttt{F}_1, \ldots, \texttt{T}_n, \texttt{F}_n\}$. In the first stage, there are $n$ online nodes of degree two, with the $i$\textsuperscript{th} online node neighboring $\texttt{T}_i$ and $\texttt{F}_i$. The second stage is sampled from distribution $\mathcal{D}$, which consists of $m$ nodes each arriving independently with probability $p = $. Each node has degree three and corresponds to a different clause $C$ in $\phi$; for each of the literals $l_i \in C$, if $l_i = x_i$ the node has an edge to $\texttt{F}_i$ while if $l_i = \overline{x_i}$ it has an edge edge to $\texttt{T}_i$.

    It is straightforward to use an exchange argument to argue that $\text{OPT}_{\text{on}}(G_{\phi})$ matches every first-stage node. Then, an arriving second-stage node can only be matched if the corresponding clause is satisfied in the truth assignment provided by this first-stage matching. Thus $$\text{OPT}_{\text{on}}(G_{\phi}) \le n + \text{OPT}(\phi) \cdot p$$ where $\text{OPT}(\phi)$ denotes the number of clauses satisfied in an optimal assignment. Conversely, $$\text{OPT}_{\text{on}}(G_{\phi}) \ge n + \text{OPT}(\phi) \cdot p - 2n \cdot 6p^2. $$ Indeed, consider the matching algorithm which in the first stage follows the assignment of $\text{OPT}(\phi)$, and attempts to match all second-stage arrivals which correspond to satisfied clauses. This would result in an expected gain of $n + \text{OPT}(\phi) \cdot p$ if there was no contention among arrivals in the second stage. The total number of edges an individual offline node sees realized in the second-stage that cannot be matched due to contention is upper bounded by $$\mathbb{E} \left[ \max \left( 0, \text{Ber}(4,p) - 1 \right) \right] \le 6p^2$$ as each variable appears in only four clauses. We can bound $\text{OPT}(\phi) \ge \tfrac{7}{8} m = \tfrac{7}{6}n$. We do not optimize constants, but note that by taking $p = (1-\sqrt{\alpha} ) \cdot 0.1$, for $\alpha$ as in \Cref{thm:max34}, a direct calculation shows it is hence hard to approximate $\text{OPT}_{\text{on}}(G_{\phi})$ within some constant factor $\beta := 1 - 0.01 ( 1 - \alpha )^2$. 
\end{proof}

We note that the same construction implies an analogous hardness-of-approximation result for the two-stage problem whose second stage is realized from a distribution $\mathcal{D}$ with polynomial-size support. Indeed, it suffices to realize the arrivals from the second stage such that each node arrives with probability $p$ and arrivals satisfy \emph{pairwise} independence. Classic constructions of pairwise-independent hash functions provide examples of such distributions over $2^{[n]}$, whose support grows polynomially in the size of $n$ for a fixed constant $p$.\footnote{For example, consider a distribution $\mathcal{D}$ defined by first picking $a$ and $b$ uniformly at random from $\mathbb{Z}_n = \{0, 1, \ldots, n-1\}$ for prime $n$, and then outputting the subset of all $i \in \mathbb{Z}_i$ for which $a + b \cdot i \text{ (mod } n \text{)} \le p \cdot n$.}
\Section{Deferred Proofs} 
\label{app:deferredproofs}
\Paragraph{Validity of LP relaxation}
\begin{claim}
    \eqref{LPon} is a relaxation of the optimum online, with an optimal value of at least $\textup{OPT}\eqref{LPon}.$
\end{claim}
\begin{proof}
Constraint \eqref{eq:lp_online_constraint} ensures that each offline node \(i\) is fractionally matched to at most one online node (across both batches); constraints \eqref{eq:lp_A} and \eqref{eq:lp_B} enforce that each online node is matched at most once in its batch; and the conditions above guarantee that only valid edges receive nonzero weight. It is simple to verify that the LP is indeed a relaxation. Let $\mathcal{M}^\theta$ denote the matching output by the optimal online algorithm under realization $\theta$, and define
\[x_{ia} = \mathbbm{1}\{(i,a) \in \mathcal{M}^\theta\} \quad \quad \text{and} \quad \quad y_{ib}^\theta = \mathbbm{1}\{(i,b) \in \mathcal{M}^\theta\}.\]
Note that $x_{ia}$ is well-defined, since the optimal online algorithm has no knowledge of the realization of $\theta$, and so selects the same first stage matching under all realizations. Moreover, the validity of Constraints \eqref{eq:lp_online_constraint}, \eqref{eq:lp_A}, \eqref{eq:lp_B}, and \eqref{eq:lp_nonneg} all follow directly from the fact that $\mathcal{M}^\theta$ is a valid matching for all realizations $\theta \in \textup{supp}(\mathcal{D})$. Finally, the zero weight condition on non-edges follows from the fact that $\mathcal{M}^\theta \subseteq E \cup E^{\theta}$, so non-edges have a value of zero by default.

Finally, the objective value of this solution is exactly the expected weight of the matching returned by the optimal online algorithm. 
\begin{align*}
    \E^{\theta}[|\mathcal{M}^\theta|] 
        &= \E^{\theta}\left[\sum_{i \in I} \left(\sum_{a \in B_1} \mathbbm{1}\{(i,a) \in \mathcal{M}^\theta\} + \sum_{b \in B_2} \mathbbm{1}\{(i,b) \in \mathcal{M}^\theta\}  \right) \right] \\
        &= \sum_{i \in I} \sum_{a \in B_1} \mathbbm{1}\{(i,a) \in \mathcal{M}^\theta\} + \E^{\theta}\left[\sum_{i \in I}\sum_{b \in B_2} \mathbbm{1}\{(i,b) \in \mathcal{M}^\theta\} \right]  \\
        &= \sum_{a \in B_1} \sum_{i \in I} x_{ia} + \mathbb{E}_{\theta}\Biggl[\sum_{b \in B_2} \sum_{i \in I} y^\theta_{ib}\Biggr].
\end{align*}
\end{proof}

\Paragraph{Negative dependence of node availabilities}
\depRoundNA*
\begin{proof}
Let $X_v^{(k)}$ be the fractional degree of node $v \in R$ after applying $k$ steps of dependent rounding to the initial matching $x$. We show by induction on $k$ that the random variables $(X_v^{(k)})_{v \in R}$ are NA. Recall it suffices to show that for any two disjoint index sets $S, T \subseteq R$ and any two functions $f: [0,1]^S \to \mathbb{R}$ and $g: [0,1]^T \to \mathbb{R}$ that are both non-decreasing in each coordinate, it holds that
\[\mathbb{E}[f(X^{(k)}_S) \cdot g(X^{(k)}_T)] \le \mathbb{E}[f( X^{(k)}_S)] \cdot \mathbb{E}[g(X^{(k)}_T)].\]
Note that upon termination, the node degrees will be exactly the collection $(\overline{A_v})_{v \in R}$.

The case where $k = 0$ is immediate, as $(X_v^{(0)})_{v \in R}$ is deterministic. Next, consider the $(k+1)^{\textup{st}}$ iteration of dependent rounding. If we identify a cycle $P$ in the first step then $X_v^{(k+1)} = X_v^{(k)}$ for all $v$, and we invoke the inductive hypothesis. If we instead identify a maximal path $P$, for fixed disjoint $S, T \subseteq R$ we consider cases based on the endpoints of $P$. For each, let $p = \beta \cdot (\alpha + \beta)^{-1}$ be the probability of an $\alpha$-adjustment and $1-p$ be the probability of a $\beta$-adjustment.
\begin{enumerate}
    \item \textit{Neither endpoint of $P$ is in $S \cup T$}.
    The conclusion is immediate as $X_i^{(k+1)} = X_i^{(k)}$ for all $i \in S \cup T$.
    \item \textit{One endpoint of $P$ is in $S \cup T$}. Call it $i$, and assume without loss of generality that \(i\in S\) and $i$ is incident to an edge in $M_1$.  For every \(s\in S\setminus\{i\}\) and every \(t\in T\) the rounding step leaves  
    \[
    X_s^{(k+1)}=X_s^{(k)} 
    \quad \text{and} \quad
    X_t^{(k+1)}=X_t^{(k)}.
    \] 
    To handle node $i$, define the monotone increasing function\[h(X_S^{(k)}) := p \cdot f(X_i^{(k)} + \alpha, X^{(k)}_{S \setminus \{i\}}) + (1-p) \cdot f(X_i^{(k)} - \beta, X^{(k)}_{S \setminus \{i\}}).\]  Because $X_i^{(k+1)} = X_i^{(k)} + \alpha$ with probability $p$ and $X_i^{(k+1)} = X_i^{(k)} - \beta$ with probability $1-p$,
    \begin{align*}
    \E[f(X_S^{(k+1)}) \cdot g(X_T^{(k+1)})]
        &= \E[h(X_S^{(k)}) \cdot g(X_T^{(k)})] \\
        &\leq \E[h(X_S^{(k)})] \cdot \E[g(X_T^{(k)})] \\
        &= \E[f(X_S^{(k+1)})] \cdot \E[g(X_T^{(k+1)})],
    \end{align*}  
    where the inequality follows from the inductive hypothesis.
    \item \textit{Both endpoints of $P$ are in $S$ (or $T$)}. Call these endpoints $i$ and $j$, and assume without loss of generality $i, j \in S$. Then $X_s^{(k+1)} = X_s^{(k)}$ for all $s \in S \setminus \{i, j\}$, $X_t^{(k+1)} - X_t^{(k)}$ for all $t \in T$, and $(X^{(k)}_i,X^{(k)}_{j})$ is sent to either $(X^{(k)}_i+\alpha,X^{(k)}_{j}-\alpha)$ with probability \ $p$ or $(X^{(k)}_i-\beta,X^{(k)}_{j}+\beta)$ with probability\ $1-p$.  Set
\[
h(X_S^{(k)}) 
= p \cdot \,f\bigl(X_i^{(k)}+\alpha,X_{j}^{(k)}-\alpha,X^{(k)}_{S\setminus\{i,j\}}\bigr)
+ (1-p) \cdot \,f\bigl(X_i^{(k)}-\beta,X_{j}^{(k)}+\beta,X^{(k)}_{S\setminus\{i,j\}}\bigr).
\]
By monotonicity of $f$ in each coordinate we again have that $h$ is non-decreasing. The same argument as in the previous case yields $\E[f(X_S^{(k+1)}) \cdot g(X_T^{(k+1)})]\le\E[f(X_S^{(k+1)})]\E[g(X_T^{(k+1)})]$.\\

\item \textit{One endpoint of $P$ is in each of $S$ and $T$}. Say $i \in S$ and $j \in T$, and assume without loss of generality that $i$ is adjacent to an edge in $M_1$ while $j$ is adjacent to an edge in $M_2$ (note the matchings must be different, as any path between $i$ and $j$ has even length). All other fractional degrees remain fixed, and $(X^{(k)}_i,X^{(k)}_j)$ is sent to either $(X^{(k)}_i+\alpha,X^{(k)}_j-\alpha)$ with probability $p$ or $(X_i-\beta,X_j+\beta)$ with probability $1-p$. Write
\[
\begin{array}{cc}
     f^+ = f\bigl(X_i^{(k)}+\alpha,X^{(k)}_{S\setminus\{i\}}\bigr) \qquad & \qquad f^- = f\bigl(X_i^{(k)}-\beta,X^{(k)}_{S\setminus\{i\}}\bigr) \\[1em]
     g^- = g\bigl(X_j^{(k)}-\alpha,X^{(k)}_{T\setminus\{j\}}\bigr) \qquad & \qquad g^+ = g\bigl(X_j^{(k)}+\beta,X^{(k)}_{T\setminus\{j\}}\bigr).
\end{array}
\]
By the inductive hypothesis,
\begin{align*}
\E[f(X_S^{(k+1)}) \cdot g(X_T^{(k+1)})] 
    &= p \cdot \E[f^+ \cdot g^-] + (1-p) \cdot \E[f^- \cdot g^+] \\
    &\leq p \cdot \E[f^+] \cdot \E[g^-] + (1-p) \cdot \E[f^-] \cdot \E[g^+].
\end{align*}
At the same time,
\[
\E[f(X_S^{(k+1)})]
= p \cdot \E[f^+] + (1-p) \cdot \E[f^-]\quad \text{and} \quad
\E[g(X_T^{(k+1)})]
= p \cdot \E[g^-] + (1-p) \cdot \E[g^+].
\]
Putting everything together, we have
\begin{align*}
\E[f(X_S^{(k+1)})] &\cdot \E[g(X_T^{(k+1)})] - \E[f(X_S^{(k+1)}) \cdot g(X_T^{(k+1)})] \\
&\geq \bigl(p \cdot \E[f^+]+(1-p) \cdot \E[f^-]\bigr)\bigl(p\cdot \E[g^-]+(1-p)\cdot \E[g^+]\bigr)\\
&\qquad - \bigl(p\cdot \E[f^+]\cdot\E[g^-]+(1-p)\cdot \E[f^-]\cdot\E[g^+]\bigr)\\
&= p\cdot (1-p) \cdot \Bigl(\E[f^+]\cdot\E[g^+] + \E[f^-]\cdot \E[g^-] - \E[f^+] \cdot \E[g^-] - \E[f^-] \cdot \E[g^+]\Bigr)\\
&= p \cdot (1-p) \cdot\bigl(\E[f^+]-\E[f^-]\bigr)\bigl(\E[g^+]-\E[g^-]\bigr).
\end{align*}

This quantity is non-negative since \(f\) and \(g\) are non-decreasing.
\end{enumerate}
For fixed disjoint $S, T \subseteq R$, these four cases cover all possible configurations, yielding the desired result. 
\end{proof}

\Paragraph{Edge-weighted analysis}
\factCRSineq*
\begin{proof}
    It suffices to show that
    \[c \frac{x^{x+1}}{(x+1)^{x+1}} \leq \frac{(x-1)^x}{x^x} \qquad \text{or equivalently} \qquad f(x) := \frac{x+1}{x}\left(1 - \frac{1}{x^2}\right)^x \geq c.\]
   First note that $f(3) = \frac{4}{3} \cdot \left(\frac{8}{9}\right)^3 \approx 0.936 \geq c$. Then, the log-derivative of $f$ is
    \begin{align*}
        \frac{d}{dx} \log f(x) = -\frac{1}{x(x+1)} +\log\left(1 - \frac{1}{x^2}\right) + \frac{2}{x^2 - 1}.
    \end{align*}
    Applying the fact that $\log(1-y) \geq -y-y^2$ for $y \in [0,1]$ with $y = \frac{1}{x^2}$ yields
\begin{align*}
     \frac{d}{dx} \log f(x) \geq -\frac{1}{x(x+1)} - \frac{1}{x^2} - \frac{1}{x^4} + \frac{2}{x^2 - 1} = \frac{x^3+1}{x^4(x^2-1)} \geq 0.
\end{align*}
Thus, $f$ is increasing for $x \geq 3$, and the desired result follows immediately.
\end{proof}

\begin{lemma}
Let $\{y_i\}_{i=1}^n$ be non-negative numbers with $\sum_i y_i \le 1$. Assume there is a distribution over active elements satisfying the condition that for any subset $S \subseteq [n]$, $$\Pr[ \ge 1 \textup{ elt. in } S \textup{ is active}] \ge c \cdot \sum_{i \in S} y_i.$$ Then, there exists a monotone contention resolution scheme, which upon observing the active elements, selects at most one such that $\Pr[i \textup{ selected}] = c \cdot y_i.$
\end{lemma}

\begin{proof}
    We follow ideas due to \cite{chekuri2014submodular}. Let $\Psi$ denote the set of monotone deterministic contention resolution schemes. Equivalently, $\Psi$ consists of all mappings $\phi : 2^{[n]} \rightarrow [n] \cup \{\emptyset\}$ such that $\phi(A) \in A \cup \{\emptyset\}$ for any $A \subseteq [n]$, and $\phi$ is monotone (i.e., if $\phi(A) = i$ for $i \in A$, and $i \in B \subseteq A$, then $\phi(B) = i$).
    Let $(\lambda_\phi)_{\phi \in \Psi}$ be a probability distribution over these deterministic schemes. Such a distribution defines a randomized monotone contention resolution scheme.
    Going forward, let $A$ denote the random subset of active elements. Let $q_{i,\phi} = \E_A[\1(\phi(A)=i)]$ be the probability that item $i$ is selected by the deterministic scheme $\phi$, where the expectation is over the active elements. 
    The problem of finding the maximum $c$ such that $\Pr[i \text{ selected}] \ge c \cdot y_i$ for all $i \in [n]$, restricted to a distribution of schemes $\phi \in \Psi$, can be formulated as the following linear program (LP):
    \begin{align*}
        \text{(LP)} \quad \max \quad & c \\
        \text{s.t.} \quad & \sum_{\phi \in \Psi} \lambda_{\phi} q_{i,\phi} \ge c \cdot y_i, \quad \forall i \in [n] \\
        & \sum_{\phi \in \Psi} \lambda_{\phi} = 1 \\
        & \lambda_{\phi} \ge 0, \quad \forall \phi \in \Psi     \end{align*}
    The dual of this LP is:
    \begin{align*}
        \text{(DP)} \quad \min \quad & \mu \\
        \text{s.t.} \quad & \sum_{i=1}^n q_{i,\phi} z_i \le \mu, \quad \forall \phi \in \Psi \\
        & \sum_{i=1}^n y_i z_i \ge 1 \\
        & z_i \ge 0, \quad \forall i \in [n]
    \end{align*}
    We will show that the objective of the dual 
    is at least $c$, the constant from the lemma's hypothesis (i.e., $\Pr[A \cap S \ne \emptyset] \ge c \sum_{j \in S} y_j$).
    To do this, we show that for any feasible dual solution $(z_1, \dots, z_n, \mu)$, we have $\mu \ge c$.
    The dual objective can be written as $$\min \mu = \min_{\substack{z_i \ge 0 \\ \sum y_i z_i \ge 1}} \left( \max_{\phi \in \Psi} \sum_{i=1}^n q_{i,\phi} z_i \right).$$
    We can write $$\sum_{i=1}^n q_{i,\phi} z_i = \E_A\left[\sum_{i=1}^n z_i \1(\phi(A)=i)\right]= \E_A[z_{\phi(A)}],$$ where we use that $\phi(A)$ is either a single item $j \in A$ or $\emptyset$, and use the convention $z_\emptyset = 0$. Note that the optimal deterministic monotone scheme, for a given active set $A$ and weights $z_i \ge 0$, is to select an item $j \in A$ that maximizes $z_j$. If $A = \emptyset$ or all $z_j=0$ for $j \in A$, it selects nothing (value 0). So, $\max_{\phi \in \Psi} \E_A[z_{\phi(A)}] = \E_A[\max_{j \in A} z_j]$ (where $\max_{j \in \emptyset} z_j = 0$).
    Let $z_1 \ge z_2 \ge \dots \ge z_n \ge 0$ be the weights, sorted without loss of generality. Let $z_{n+1}=0$. We then have
    $$ \E_A \left[ \max_{j \in A} z_j \right] = \sum_{k=1}^n (z_k - z_{k+1}) \Pr[\max_{j \in A} z_j \ge z_k]. $$
    The event $\{\max_{j \in A} z_j \ge z_k\}$ is equivalent to the event that $A \cap S_k \ne \emptyset$, where $S_k = \{1, 2, \dots, k\}$ (indices corresponding to sorted $z_j$'s).
    Using the hypothesis of the lemma, $\Pr[A \cap S_k \ne \emptyset] \ge c \sum_{j \in S_k} y_j$.
    So,
    \begin{align*} \E_A \left[ \max_{j \in A} z_j \right] &\ge \sum_{k=1}^n (z_k - z_{k+1}) \left(c \sum_{j=1}^k y_j\right) = c \sum_{k=1}^n (z_k - z_{k+1}) \sum_{j=1}^k y_j = c \sum_{j=1}^n z_j y_j \end{align*}
    Since $z_i$ must satisfy $\sum_{j=1}^n y_j z_j \ge 1$, we have $\E_A[\max_{j \in A} z_j] \ge c \cdot 1 = c$.
    Therefore, the optimal value of (DP) is at least $c$. By strong duality, there hence exists a randomized monotone CR scheme $(\lambda_\phi)_{\phi \in \Psi}$ such that $\Pr[i \text{ selected}] = \sum_{\phi \in \Psi} \lambda_\phi q_{i,\phi} \ge c \cdot y_i$ for all $i \in [n]$.

    To conclude, we show that we can achieve exact equality here without violating monotonicity. Let $P_i = \Pr[i \text{ selected}]$. We have $P_i \ge c y_i$. To achieve equality, we can modify the scheme by adding a final step such that if the original scheme selects $i$ (which happens with probability $P_i$), we make the final selection of $i$ with probability $cy_i/P_i \le 1$. This ``thinning'' procedure transforms a randomized monotone CR scheme into another randomized monotone CR scheme. Specifically, if the original scheme is a convex combination of deterministic monotone CR schemes, the resulting scheme also is.
    Thus, there exists a monotone contention resolution scheme (which is a convex combination of deterministic monotone CR schemes) such that $\Pr[i \text{ selected}] = c \cdot y_i$ for all $i \in [n]$.
\end{proof}

\claimNAprelim*
\begin{proof}
Fix disjoint \(I,J\subseteq\{0,1,\dots,n\}\) and increasing $f(\cdot)$ and $g(\cdot)$ on these respective domains. If \(0\notin I\cup J\), then \(f\) and \(g\) depend only on the \(X\)’s, so 
\[
\E\bigl[f(\vec X_I)\,g(\vec X_J)\bigr]\le\E\bigl[f(\vec X_I)\bigr]\,\E\bigl[g(\vec X_J)\bigr]
\]
by negative association of \((X_1,\dots,X_n)\). Otherwise assume \(0\in I\) and write \(I' = I\setminus\{0\}\). For each \(b\in\{0,1\}\), define 
\[
f_b\bigl(\vec X_{I'}\bigr) = f\bigl((b,\vec X_{I'})\bigr).
\]
By the law of total expectation over the support \(\mathcal{X}\) of \(X_0\), 
\[
\E\bigl[f(\vec X_I)\,g(\vec X_J)\bigr]
=\sum_{x \in \mathcal{X}} \Pr\bigl[X_0 = x\bigr] \cdot \E\bigl[f(\vec X_I)\,g(\vec X_J)\big|X_0 = x\bigr].
\]
Given \(X_0 = x\), the conditional distribution of \(Y\) is \(\mathrm{Ber}(x)\) and is independent of \(\vec X_{I'}\) and \(\vec X_J\). Hence
\[
\E\bigl[f(\vec X_I)\,g(\vec X_J)\big|X_0 = x \bigr]
= \mathbb{E} \bigl[ (1 - x)\,f_0\bigl(\vec X_{I'}\bigr)\,g\bigl(\vec X_J\bigr)+x\,f_1\bigl(\vec X_{I'}\bigr)\,g\bigl(\vec X_J\bigr) \bigr].
\]
Thus
\[
\E\bigl[f(\vec X_I)\,g(\vec X_J)\bigr]
=\E\Bigl[\,(1 - X_0)\,f_0\bigl(\vec X_{I'}\bigr)\,g\bigl(\vec X_J\bigr)
+X_0\,f_1\bigl(\vec X_{I'}\bigr)\,g\bigl(\vec X_J\bigr)\Bigr].
\]
Define 
\[
h\bigl(X_0,\vec X_{I'}\bigr)
=(1 - X_0)\,f_0\bigl(\vec X_{I'}\bigr)+X_0\,f_1\bigl(\vec X_{I'}\bigr).
\]
Since \(f_1\ge f_0\), the function \(h\) is increasing in \(X_0\); it is also increasing in each \(X_i\) for \(i\in I'\).  Likewise \(g\) is increasing in each \(X_j\) for \(j\in J\) by assumption.  Because \(\{0\}\cup I'\) and \(J\) are disjoint, the negative association of \((X_0,X_1,\dots,X_n)\) implies
\[
\E\bigl[h\bigl(X_0,\vec X_{I'}\bigr)\,g\bigl(\vec X_J\bigr)\bigr]
\le\E\bigl[h\bigl(X_0,\vec X_{I'}\bigr)\bigr]\E\bigl[g\bigl(\vec X_J\bigr)\bigr].
\]
Noting that 
\[
\E\bigl[h\bigl(X_0,\vec X_{I'}\bigr)\bigr]
=\E\Bigl[(1 - X_0)\,f_0\bigl(\vec X_{I'}\bigr)+X_0\,f_1\bigl(\vec X_{I'}\bigr)\Bigr]
=\E\bigl[f(\vec X_I)\bigr],
\]
we conclude
\[
\E\bigl[f(\vec X_I)\,g(\vec X_J)\bigr]
\le\E\bigl[f(\vec X_I)\bigr]\E\bigl[g(\vec X_J)\bigr]. \qedhere
\]
\end{proof}

\Section{Approximation to Optimum Offline} \label{app:offline}
An earlier result from \cite{feng2021two} shows that the best possible approximation to the optimum offline for stochastic two-stage matching is $\nicefrac{3}{4}$. We give a simple two-stage rounding algorithm of a natural LP relaxation that achieves the same bound. Although this a weaker result than that of \cite{feng2021two}, who provide a $\nicefrac{3}{4}$-competitive algorithm for adversarial input, we include it here as a simple alternative proof. 

\Paragraph{LP relaxation of \(\OptOff\)}\label{subsec: LP-relax-exp}

Let $x_{ia}$ denote the fractional matching value for each first-stage node pair $i \in I$ and $a \in B_1$. Similarly, for all realizations $\theta \in \Theta$, let $y^\theta_{ib}$ denote the fractional matching value for each second-stage node pair $i \in I$ and $b \in B_2$. We enforce that non-edges are not included in the matching, i.e., 
\[
x_{ia}=0 \quad \text{if } (i,a)\notin E,\qquad \text{and} \qquad y^\theta_{ib}=0 \quad \text{if } (i,b)\notin E_\theta.
\]
An LP relaxation for \(\OptOff\) is given by
\begin{align}
\max \quad & \sum_{i \in I}\sum_{a \in B_1} x_{ia} + \sum_{i \in I} \sum_{b \in B_2} \mathbb{E}_{\theta \sim \mathcal{D}}\Bigl[y^\theta_{ib}\Bigr] \tag{$\text{LP}_{\text{off}}$}\\[1mm]
\text{s.t.} \quad & \sum_{a \in B_1} x_{ia} + \sum_{b \in B_2} \mathbb{E}_{\theta \sim \mathcal{D}}\Bigl[y^\theta_{ib}\Bigr] \le 1, \quad \forall i \in I, \label{eq:off_lp_I_agg}\\[1mm]
& \sum_{b \in B} y^\theta_{ib} \le 1, \quad \forall i \in I,  \forall \theta \in \Theta, \label{eq:off_lp_I_cap}\\[1mm]
& \sum_{i \in I} x_{ia} \le 1, \quad \forall a \in B_1, \label{eq:off_lp_A_cap}\\[1mm]
& \sum_{i \in I} y^\theta_{ib} \le 1, \quad \forall b \in B_2, \forall \theta \in \Theta, \label{eq:off_lp_B_cap}\\[1mm]
& x_{ia} \ge 0, \quad \forall i \in I, a \in B_1,\qquad y^\theta_{ib} \ge 0, \quad \forall i \in I, b \in B_2, \forall \theta \in \Theta. \label{eq:lp_nonneg}
\end{align}
Constraint \eqref{eq:off_lp_I_agg} ensures that, in expectation, the fractional degree of each offline node \(i\) at most one online node across both batches. Constraints \eqref{eq:off_lp_A_cap} and \eqref{eq:off_lp_B_cap} enforce that each online node is matched at most once in its batch; and the conditions above guarantee that only valid edges receive nonzero weight. 

It is simple to verify that this LP is indeed a relaxation. For each realization $\theta \in \Theta$, let $\mathcal{M}^\theta$ be the integral max-cardinality matching returned by the optimal offline algorithm. Define a solution $(x,y)$ to the LP via
\[
x_{ia} = \Pr_{\theta \sim \mathcal{D}}\bigl[(i,a) \in \mathcal{M}^\theta\bigr] \quad\quad \text{and} \quad\quad y^\theta_{ib} = \mathbbm{1}\{(i,b) \in \mathcal{M}^\theta\} \quad \quad \forall a \in B_1, b \in B_2, i \in I, \theta \in \Theta.
\]
It is clear that the non-negativity and non-edge constraints hold under this construction. Because $\mathcal{M}^\theta$ is a valid matching, the degree of each offline node cannot exceed 1:
\[\sum_{a\in B_1}\mathbbm{1}{\{(i,a)\in\mathcal{M}^\theta\}} + \sum_{b\in B_2}\mathbbm{1}{\{(i,b)\in\mathcal{M}^\theta\}} \le 1 \qquad \text{for all $i \in I$}.\]
It immediately follows that constraint \eqref{eq:off_lp_I_cap} is satisfied, and taking an expectation of both sides with respect to $\theta$ shows that constraint \eqref{eq:off_lp_I_agg} also holds. Similarly, each online node has degree at most 1 under $\mathcal{M}^\theta$:
\[\sum_{i \in I} \mathbbm{1}\{(i,a) \in \mathcal{M}^\theta\} \leq 1 \quad \forall a \in B_1 \qquad \text{and} \qquad \sum_{i \in I} \mathbbm{1}\{(i,b) \in \mathcal{M}^\theta\} \leq 1 \quad \forall b \in B_2.\]
Taking an expectation of both sides of the first inequality with respect to $\theta$ validates constraint \eqref{eq:off_lp_A_cap}, and the second gives constraint \eqref{eq:off_lp_B_cap} directly. Finally, the LP objective value corresponding to the solution $(x,y)$ is exactly $\OptOff$, the expected number of edges matched by the optimal offline algorithm.
\[
\sum_{a \in B_1} \sum_{i \in I} x_{ia} + \mathbb{E}_{\theta \sim \mathcal{D}} \left[ \sum_{b \in B_2} \sum_{i \in I} y^\theta_{ib} \right] = \E_{\theta \sim \mathcal{D}}\biggl[\sum_{a \in B_1} \sum_{i \in I} \mathbbm{1}\{(i,a) \in \mathcal{M}^\theta\} +  \sum_{b \in B_2} \sum_{i \in I} \mathbbm{1}\{(i,b) \in \mathcal{M}^\theta\} \biggr]
\]
Hence, we have constructed a feasible solution to the LP with objective value $\OptOff$, implying that the LP optimum is at least $\OptOff$.

\Paragraph{Two-stage online rounding algorithm}

Given an optimal fractional solution \((x,y)\) to the LP, our online rounding procedure proceeds in two stages. In batch \(B_1\), we round the fractional solution \(\vec{x}\) to obtain an integral matching \(M_1\). In batch \(B_2\), upon observing a realization \(\theta\) of the online nodes (with corresponding edge set \(E_\theta\)), we round the fractional solution \(\vec{y}^\theta\) on the offline nodes that remain unmatched from \(M_1\) to obtain an integral matching \(M_2^\theta\). In cases of conflict, matches from batch \(B_1\) take precedence.

\begin{algorithm}[H]
\caption{Two-Stage Online Rounding}\label{alg:rounding}
\begin{algorithmic}[1]
\STATE \textbf{Input:} Optimal (fractional) solution \((\vec{x},\vec{y})\) from the LP 
\STATE Upon arrival of batch $B_1$:
\STATE \quad round \(\vec{x}\) using a dependent rounding scheme to obtain matching \(M_1\).
\STATE Upon arrival of batch \(B_2\) with realization \(\theta\)
\STATE \quad round \(\vec{y}^\theta\) via dependent rounding to obtain matching \(M_2^\theta\).
\STATE \textbf{Return:} \(M = M_1 \cup M_2^\theta\)
\end{algorithmic}
\end{algorithm}

\begin{theorem}\label{thm:off-approx}
\Cref{alg:rounding} achieves a competitive ratio of $\nicefrac{3}{4}$ relative to the LP optimum—and hence at least \(\frac{3}{4}\) of \(\OptOff\).
\end{theorem}
\begin{proof}
For each offline node \(i\in I\), define the matching contributions $x_i := \sum_{a \in B_1} x_{ia}$ and $y_i := \E[\sum_{b \in B_2} y^\theta_{ib}]$ from the first and second batch, respectively. Because the rounding procedure preserves marginal probabilities, the probability that \(i\) is matched in batch \(B_1\) is exactly \(x_i\). Moreover, if \(i\) is unmatched in the first batch (which happens with probability \(1-x_i\)), then \(i\) is matched in batch \(B_2\) with probability at least \(y_i\). Thus, the overall probability that offline node \(i \in I\) is matched by \Cref{alg:rounding} is at least
\[
x_i + (1-x_i)y_i = x_i + y_i - x_i y_i.
\]
By constraint \eqref{eq:off_lp_I_agg} of the LP, we have $x_i + y_i \leq 1$. Under this condition and non-negativity,
\[\min_{x \geq 0, \; y \geq 0, \; x + y \leq 1} \frac{x + y - xy}{x + y} = \frac{3}{4},\]
with the minimum achieved by the solution $x = y = 1/2$. Applying this bound to all offline nodes yields the result. We note that the same result holds for the vertex-weighted problem, when updating the LP relaxation accordingly. 
\end{proof}

\end{document}